\documentclass{sig-alternate}

\usepackage{theorem}

\usepackage{dsfont}
\makeatletter
\def\@copyrightspace{\relax}
\makeatother

\usepackage[ruled,vlined]{algorithm2e}
\renewcommand{\algorithmcfname}{ALGORITHM}
\SetAlFnt{\small}
\SetAlCapFnt{\small}
\SetAlCapNameFnt{\small}
\SetAlCapHSkip{0pt}
\IncMargin{+1em}

\newcommand{\E}{\mathbb{E}}

\newtheorem{theorem}{Theorem}[section]
\newtheorem{lemma}[theorem]{Lemma}

\newtheorem{corollary}[theorem]{Corollary}
\newtheorem{defn}[theorem]{Definition}

\def\blackslug{\rule{2.5mm}{2.5mm}}
\def\qed{\hfill\blackslug}

\newcommand{\bg}[1]{\medskip\noindent{\bf #1}}

\newenvironment{proofof}[1]{\bg{Proof of #1 : }}{\qed}

\newenvironment{proofsketch}{\noindent{\textsc{Proof sketch}.}}{\qed}

\newcommand{\comment}[1]{}

\newcommand{\R}{\ensuremath{\mathbb R}}

\newcommand{\N}{\ensuremath{\mathbb N}}

\newcommand{\M}{\ensuremath{\mathcal M}}
\newcommand{\Mj}{\ensuremath{\mathcal M}_j}

\newcommand{\X}{\ensuremath{\mathcal X}}
\newcommand{\Xsi}{\ensuremath{\mathcal X}_i}
\newcommand{\Xsj}{\ensuremath{\mathcal X}^j}
\newcommand{\Xsij}{\ensuremath{\mathcal X}_i^j}
\newcommand{\x}{\ensuremath{x}}
\newcommand{\xsi}{\ensuremath{x_i}}

\newcommand{\xsij}{\ensuremath{x_i^j}}

\newcommand{\V}{\ensuremath{\mathcal V}}
\newcommand{\Vi}{\ensuremath{\mathcal V}_i}

\newcommand{\Vij}{\ensuremath{\mathcal V}_i^j}
\newcommand{\val}{\ensuremath{v}}
\newcommand{\vali}{\ensuremath{v_i}}

\newcommand{\A}{\ensuremath{\mathcal A}}
\newcommand{\Ai}{\ensuremath{\mathcal A}_i}
\newcommand{\Aj}{\ensuremath{\mathcal A}^j}
\newcommand{\Aij}{\ensuremath{\mathcal A}_i^j}
\newcommand{\ac}{\ensuremath{a}}
\newcommand{\aci}{\ensuremath{a_i}}
\newcommand{\acj}{\ensuremath{a^j}}
\newcommand{\acij}{\ensuremath{a_i^j}}
\newcommand{\al}{\ensuremath{\mathbf a}}

\newcommand{\st}{\ensuremath{\mathbf s}}

\newcommand{\El}{\ensuremath{\mathcal L}}

\newcommand{\opt}{\text{\textsc{Opt}} }

\begin{document}

\title{Composable and Efficient Mechanisms}

\numberofauthors{2}

\author{
\alignauthor
Vasilis Syrgkanis\titlenote{Supported in part by ONR grant N00014-98-1-0589,
NSF grant CCF-0729006 and a Simons Graduate Fellowship.}\\
       \affaddr{Cornell University}\\
       \email{vasilis@cs.cornell.edu}
\alignauthor
\'{E}va Tardos\titlenote{Supported in part by NSF grants CCF-0910940 and CCF-1215994, ONR grant N00014-08-1-0031,  a Yahoo!~Research Alliance Grant, and
a Google Research Grant.\ \\ \ \\ \ \\}\\
       \affaddr{Cornell University}\\
       \email{ eva@cs.cornell.edu}
}

\maketitle
\begin{abstract}
We initiate the study of efficient mechanism design with guaranteed good properties even when players participate in multiple different mechanisms simultaneously or sequentially. We  define the class of smooth mechanisms, related to smooth games defined by Roughgarden, that can be thought of as mechanisms that generate approximately market clearing prices. We show that smooth mechanisms result in high quality outcome in equilibrium both in the full information setting and in the Bayesian setting with uncertainty about participants, as well as in learning outcomes. Our main result is to show that such mechanisms compose well: smoothness locally at each mechanism implies efficiency globally.

For mechanisms where good performance requires that bidders do not bid above their value, we identify the notion of a weakly smooth mechanism. Weakly smooth mechanisms, such as the Vickrey auction, are approximately efficient under the no-overbidding assumption. Similar to smooth mechanisms, weakly smooth mechanisms behave well in composition, and have high quality outcome in equilibrium (assuming no overbidding) both in the full information setting and in the Bayesian setting, as well as in learning outcomes.

In most of the paper we assume participants have quasi-linear valuations. We also extend some of our results to settings where participants have budget constraints.
\end{abstract}




\section{Introduction}
The goal of our paper is to initiate the study of efficient mechanism design with guaranteed good properties even when players participate in multiple different mechanisms either simultaneously or sequentially. In most markets, (e.g. online markets) people participate in various mechanisms and the value of each player overall is a complex function of their outcomes.
Predominantly, these mechanisms are run by different principals (e.g. different sellers on eBay or different ad-exchange platforms) and coordinating them to run a single combined mechanism is infeasible or impractical. The goal of this paper is to develop a theory of how to design mechanisms so that the efficiency guarantees for a single mechanism (when studied in isolation) carry over to the same or approximately the same guarantees for a market composed of such mechanisms. The key question considered in this paper can be summarized as follows:
\advance\leftmargini -.5em
\begin{quotation}
{\em What properties of local mechanisms guarantee  global efficiency
in a market composed of such mechanisms?}
\end{quotation}
Mechanism design is a subject with a long and distinguished history aiming to design games that produce a certain desired outcome (such as revenue or social welfare maximization) in equilibrium. However, traditional mechanism design considered such mechanisms only in isolation, an assumption not so realistic in online markets, where players can cover their needs through multiple different mechanisms.
Mechanism design has mostly
focused on truthful mechanisms, where players participate by revealing their true preferences to the mechanism.
In an environment with several auctions running simultaneously or sequentially, truthfulness of each individual auction loses its appeal, as the global mechanism is no longer truthful, even if each individual part is.
The literature's focus on truthful mechanisms is based on the revelation principle, showing that if there are better non-truthful solutions, the mechanism designer can run this alternate solution on the players' behalf. However, the revelation principle is limited to mechanisms running in isolation: with multiple mechanisms run by different parties, there is no global coordinator to implement the solution. Requiring global coordination between mechanisms is not viable 
and could lead to complicated coordination problems, such as agreeing on ways to divide up the global revenue.

The online market setting introduces new desiderata for designing mechanisms. 
Typical mechanisms used in practice are extremely simple, and not truthful. The Internet environment allows for running millions of auctions, which necessitates the use of very simple and intuitive auction schemes.
Second, we cannot 
assume that the designer knows all parameters of the environment at the design phase.
Most mechanisms in online markets run in a dynamic environment and constantly adapting the mechanism is infeasible.
Third, participants of such a dynamic and complex setting are bound to use learning strategies. Therefore, a mechanism should have good properties even under learning behavior. Last, we cannot expect 
the participants to know all the parameters of the game (e.g. valuations of opponents). Therefore, the mechanism should also be robust with respect to informational assumptions and should be approximately efficient, independent of the distribution of  valuations.

\paragraph{Our Results} We define the notion of a {\em$(\lambda,\mu)$-smooth mechanism} and show that {\em smooth mechanisms} possess all the aforementioned desired properties of composability and robustness under learning behavior and incomplete information.
If a mechanism has the property that in any outcome, any participant can change her bid to receive her allocation of choice
by paying the price paid at the current outcome, then the equilibrium outcome and prices are {\em market clearing},
implying that the outcome is socially optimal. Smooth mechanisms satisfy an approximate analog of this,
requiring the property only in aggregate and only approximately (both in the value of the outcome achieved by the deviating bid and in the price paid), but not allowing the deviating bid to depend on the current actions of other players; a property crucial for the efficiency results described next. Our notion of smoothness is focused on mechanisms where players have quasilinear utilities and is
related to the notion of smooth games introduced by Roughgarden \cite{Roughgarden2009}.


%
\textbf{Smooth Mechanisms and Efficiency.} We show that a $(\lambda,\mu)$-smooth mechanism achieves at least a $\frac{\lambda}{\max(1,\mu)}$ fraction of the maximum possible social welfare in the full information setting. This is true in all correlated equilibria of the game and thereby no-internal-regret learning outcomes \cite{Blum}. We show that this result extends to the Bayesian setting with uncertainty about participants. This extension theorem strengthens the results of Roughgarden \cite{Roughgarden2012} and Syrgkanis \cite{Syrgkanis2012} who showed a similar extension theorem requiring a complex smoothness condition involving multiple types which additionally couldn't capture sequential games. Our proof
uses a bluffing technique to handle the fact that we allow the deviating action to depend on the previous action of the deviating player (needed for sequential composition).

\textbf{Complement-Free Valuations.} We develop an hierarchy of valuations
on outcomes of different mechanisms. Existing valuation hierarchies consider only valuations on sets of items.
We identify analogs of complement-free valuations across mechanisms, without making any assumption about the valuations of players' for outcomes within a mechanism. We define natural generalizations of fractionally subadditive and XOS valuations and show that these two classes are equivalent
extending the result of Feige \cite{Feige2006}.

\textbf{Composability of Smooth Mechanisms.} We show that smooth mechanisms compose well in parallel: if we run any number of  $(\lambda, \mu)$-smooth mechanisms simultaneously and players have fractionally subadditive valuations over outcomes of different mechanisms, then the global game is also a $(\lambda, \mu)$-smooth mechanism, and hence achieves a $\lambda/\max(1,\mu)$ fraction of the maximum social welfare in all correlated equilibria of the full information setting and in all mixed Bayes-Nash equilibria in the Bayesian setting.

We also show that smooth mechanisms compose well sequentially: if we run any number of $(\lambda, \mu)$-smooth mechanisms sequentially and a player's value 
is the maximum valued allocation she got among all mechanisms then the global game is also $(\lambda,\mu+1)$-smooth and thereby achieves
a $\lambda/(\mu+1)$ fraction of the optimal social welfare.


\textbf{Applications.} We show that many well-known auctions are {\em smooth} and can be analyzed in our framework. We list a few representative examples below, and note that our composition result applies when running any set of such auctions simultaneously or sequentially.
\begin{itemize}
\setlength{\itemsep}{0pt}\setlength{\parsep}{0pt}\setlength{\parskip}{0pt}
\item We show that the first price auction is $(1-1/e,1)$-smooth implying an efficiency bound of approximately $1.58$ for simultaneous first price item auctions, improving the bound of Hassidim et al \cite{Hassidim2011} and matching \cite{Syrgkanis2012}.
\item All-pay auctions, and 
a simple first price position auction are $(1/2,1)$-smooth, implying a bound 
of $2$.
\item The first price greedy combinatorial auction of Lucier and Borodin \cite{Lucier2010} based on a $c$-approximation algorithm is $(1-e^{-1/c},c)$-smooth, improving the efficiency bound of \cite{Lucier2010} from $c+O(\log(c))$ to $c+0.58$.
\item The bandwidth allocation game of Johari and Tsitsiklis \cite{Johari2004} is $(2-\sqrt{3}, 1)$-smooth, proving a somewhat weaker efficiency bound than \cite{Johari2004}, but extending the bound also to Bayesian games and learning outcomes. 
\end{itemize}

\textbf{No-overbidding.} For some mechanisms, such as the second price auction, good performance requires that bidders 
do not bid above their value. 
For such mechanisms, we identify the notion of a weakly smooth mechanism. Roughly speaking, we will require that bidders'
declared 
maximum willingness to pay doesn't exceed their valuation, and add a term to the smoothness definition using the participants maximum willingness to pay. As in the case of smooth mechanism, weakly smooth mechanisms remain weakly smooth when composed, and have high quality outcome in equilibrium (assuming no overbidding) both in the full information setting, in learning outcomes, and in the Bayesian setting.

\textbf{Budget Constraints.} The results discussed so far, assume that participants have quasi-linear valuations. The most common valuation that is not quasi-linear is when players have budget constraints. We extend our results to settings where participants have budget constraints. With budget constraints, maximizing welfare is not an achievable goal, as we cannot expect a low budget participant to be effective at maximizing her contribution to welfare. Instead, we consider the optimal ``effective welfare'' benchmark; capping the contribution of each player to the welfare by their budget. We show that all our results about efficiency for the case of simultaneous mechanisms carry over to bounds for this benchmark when players have budget constraints.

\paragraph{Related Work}
There has been a long line of research on quantifying inefficiency of equilibria starting from Koutsoupias and Papadimitriou 
\cite{Koutsoupias1999} who introduced the notion of the price of anarchy. More recently, this analysis technique has also been used to quantify the inefficiency of auction games, including games of incomplete information.
A series of papers, Bikhchandani \cite{Bikhchandani1999}, Christodoulou et al
\cite{Christodoulou2008},
Bhawalkar and Roughgarden \cite{Bhawalkar2011}, Hassidim et al
\cite{Hassidim2011}, Paes Leme et al \cite{PaesLeme2012}, Syrgkanis and Tardos \cite{Syrgkanis2012a} studied the efficiency of equilibria of
non-truthful combinatorial auctions that are based on running separate item auctions (simultaneously or sequentially) for each item.  Lucier and Borodin \cite{Lucier2010} studied Bayes-Nash Equilibria of non-truthful auctions 
based on greedy allocation algorithms. Caragiannis et al
\cite{Caragiannis2012} studied the inefficiency of Bayes-Nash equilibria of the generalized second price auction.
 All this literature can be thought of as special cases of our framework and all the proofs can be understood
as smoothness proofs giving the same or even tighter results. A recent exception is the paper by Feldman et al. \cite{Feldman2012} giving a tighter bound for simultaneous item-auctions with subadditive bidders, than what would follow from smoothness.

Roughgarden \cite{Roughgarden2009} proposed a framework, which he calls smoothness in games, and showed that a number of classical price of anarchy results (such as routing and valid utility games) can be proved using this 
framework. Further, he showed that such efficiency proofs carry over to efficiency of coarse correlated equilibria (no-regret learning outcomes).
Nadav and Roughgarden \cite{Nadav2010} give the broadest solution concept for which smoothness proofs apply.  Schoppmann and Roughgarden \cite{Roughgarden2011} extend the 
framework to games with continuous strategy spaces, providing tighter results. However, these papers consider only the full information setting
and do not capture several of the auctions described previously. Our definition of a smooth mechanism is closely related to the notion of a smooth game. If utilities of the game were always non-negative (which is not 
the case here) then a $(\lambda,\mu)$-smooth mechanism can be thought of as a $(\lambda,\mu-1$)-smooth game, but with much
weaker requirements, allowing us to capture all the auctions 
above, as well as sequential composition.

Recent papers offer extensions of the smoothness framework to incomplete information games. Lucier and Paes Leme \cite{Lucier2011} introduced the concept of semi-smoothness
(inspired by their GSP analysis), and showed that efficiency results shown via semi-smoothness extend to the incomplete information version of the game, even if the types of the players are arbitrarily correlated. However, semi-smoothness is a much more restrictive property (for instance, not satisfied by the item-bidding auctions) than just requiring that every complete information instance of the game is smooth. Recently Roughgarden \cite{Roughgarden2012} (and independently Syrgkanis \cite{Syrgkanis2012}) offered a more general such extension theorem. They show that one can prove bounds on the price of anarchy of an incomplete information game (assuming type distributions of players are independent) by restricting attention to induced complete information instances and
proving a stronger version of the smoothness property, which \cite{Syrgkanis2012} calls universal smoothness. Our extension theorem is based on simply assuming that for any choice of valuations, the induced full information game is smooth according to the standard definition of smoothness. In contrast, the stronger universal smoothness property relates utilities of players with different types in a single inequality. While, many of the known examples satisfy this stronger notion of smoothness, our extension theorem is more natural, assuming only that the underlying full information game is smooth, and does not mix player types. In addition, our smoothness is an even weaker property that allows us to capture efficiency in sequential games of incomplete information in a unified framework.

A recent survey by Pai \cite{Pai2010} highlights settings where different sellers compete by announcing mechanisms,
starting from the seminal work of McAfee \cite{McAfee1993} and focusing on revenue maximization.
Our work is in the same spirit, and aims to analyze the effect of such competition on social welfare.


\section{Composition Framework}

We will consider a setting where 
players participate in a set of 
mechanisms. We assume that players' preferences are quasilinear in money.
In this section we introduce our framework and set up the notation we need for defining mechanisms and compositions of mechanisms.

{\bf Mechanism with Quasilinear Preferences.} A mechanism design setting consists of a set of $n$ players and a set of outcomes $\X\subseteq \times_i \Xsi$, where $\Xsi$ is the set of allocations for player $i$. Each player has a valuation $v_i:\Xsi\rightarrow \R_+$ over allocations. Let $\Vi$ be the set of possible valuations of player $i$. Given an allocation $\xsi\in \Xsi$ and a payment $p_i\in \R_+$, we assume that the utility of player $i$ is:
\begin{equation}
\textstyle{u_i^{\vali}(\xsi,p_i)=\vali(\xsi)-p_i}
\end{equation}
Observe that although we assume that the outcome space is in the form of a subset of a product space, this doesn't restrict at all the space of mechanism design settings we can model, since we don't put any restriction on the structure of the subset of the product space. Hence, our model captures a range of problems, including
games where players have externalities or share a single outcome. A few of the special cases are: 
1) \emph{combinatorial auctions} where $\Xsi$ is the power set of items and $\X$ is the subset of this product space such that no item is assigned to more than one player,
2) \emph{combinatorial public projects} where $\Xsi$ is the power set of projects and $\X$ is the subset of the product space such that every coordinate is the same, 3) \emph{position auctions} where $\Xsi$ is the set of positions and $\X$ is the subset of the product space where no two coordinates are assigned the same position,
4) \emph{bandwidth allocation mechanisms} where $\Xsi$ is the portion of the bandwidth assigned to player $i$ and $\X$ is the subset such that the sum of the coordinates is at most the bandwidth capacity.
Using this product space formulation allows us to encode which part of the outcome the valuation of a player is affected by and it facilitates the formulation of valuation classes on outcome spaces as we will see in the next section.

Given a valuation space $\V = \times_i \Vi$ and 
an outcome space $\X$, a mechanism $\M$ is a triple $(\A,X,P)$, where $\A=\times_i \Ai$ is a set of actions $\Ai$ for each player $i$, $X:\A\rightarrow \X$ is an allocation function that takes an action profile $\ac$ and maps it to an outcome $\x$ and $P:\A\rightarrow \R^n_+$ is a payment
function that takes an action profile and maps it into a payment $P_i$ for each player.

We will only consider settings where each player has the option to not participate, and hence at any rational outcome gets non-negative utility in expectation over the information she doesn't have and over the randomness of the other players and the mechanisms.

{\bf The Composition Framework.} Mechanisms rarely run in isolation but rather, several mechanisms take place simultaneously and/or sequentially, and players typically have valuations that are complex functions on the outcomes of different mechanisms.

We consider the following general setting: there are $n$ bidders and $m$ mechanisms. Each mechanism $\Mj$ has its own outcome space $\Xsj$ and consists of
a triple $(\Aj,X^j,P^j)$ as described previously, i.e. $\Aj=\times_i \Aij$ is the action space, $X^j:\Aj\rightarrow \Xsj$ is the allocation function and $P^j:\Aj\rightarrow \R_+^n$ the payment function.

We assume that a player has a valuation over vectors of outcomes from the
different mechanisms: $\vali:\Xsi\rightarrow \R^+$ where $\Xsi = \times_j \Xsij$. A player's utility is still quasi-linear in this extended setting in the sense that his utility from an allocation vector $\x_i=(\x_i^1,\ldots,\x_i^m)$ and payment vector $p_i = (p_i^1,\ldots,p_i^m)$ is given by:
\begin{equation}
\textstyle u_i^{v_i}(\xsi,p_i)=v_i(x_i^1,\ldots,x_i^m)-\sum_{j=1}^{m}p_i^j
\end{equation}
We will consider both simultaneous and sequential composition of mechanisms. In the case of simultaneous composition, a player's strategy space is to report an action $\acij$ at each mechanism $j$.
In the case of sequential composition a player can base the action she submits at mechanism $j$ on the history of the submitted action profiles in previous mechanisms (alternatively we could assume that bidders observe only allocations and payments in previous mechanisms; our results are robust to such information assumptions).

The simultaneous composition of mechanisms can be viewed as a global mechanism $\M=(\A,X,P)$, where $\Ai=\times_j \Aij$, $\X=\times_j\Xsj$, $X(\ac) = (X^j(\acj))_j$ and $P(a) = \sum_j P^j(a^j)$. Sequential composition can also be viewed as a global mechanism with a more complex action space, where actions are functions of the observed history of play in earlier mechanisms.
Our goal is to give properties of the individual mechanisms that guarantee efficiency of 
the global mechanism.

{\bf Efficiency Measure.}
We will measure efficiency of an action profile $a$ in terms of social welfare
\begin{equation}
\textstyle SW^{v}(a)=\sum_i v_i(X_i(a))
\end{equation}
which is the sum of the utilities of all players and the revenue of all the mechanisms. For any valuation profile $\val\in \times_i \Vi$ there exists an optimal allocation $x^*(v)$ that maximizes $\sum_i v_i(x_i)$ over all allocations $x\in \mathcal{X}$ and we will denote with
\begin{equation}
\textstyle \opt(v)=\sum_i v_i(x_i^*(v))
\end{equation}

\section{Hierarchy of Valuations}
\label{SEC:VALUATIONS}

In order to 
to infer good properties of the global mechanism from properties of individual mechanisms, we will need to assume that player valuations have no complements across outcomes of different mechanisms. When each mechanism is a single item auction, this is captured by well-understood assumptions on valuations, such as subadditive, fractionally subadditive, submodular, etc. Here we will extend these definitions to the case when the component mechanisms have arbitrary possible outcomes, without making any assumptions on valuations within each mechanism.
Since we focus on the valuation of a specific player $i$, for notational simplicity we will drop
the index $i$ in the current section. We define classes of complement free valuations $v:\X\rightarrow \R_+$ on a
product space of allocations
$\X=\times_j\X_j$. In a composition setting, 
$\X_j$ is the set of possible allocations to player $i$ from mechanism $\Mj$.

The class of valuations that will be important in our composability theorems is that of
fractionally subadditive valuations across mechanisms, which we define as follows:
\begin{defn}[Fractionally Subadditive] A valuation is fractionally subadditive across mechanisms if
\begin{align*}
\textstyle v(x) \le \sum_{\ell} \alpha_{\ell}  v(y^{\ell})
\end{align*}
whenever each coordinate $x_{j}$ is covered in the set of solutions $y^{\ell}$, that is $\sum_{\ell: x_{j}=y^{\ell}_{j}} \alpha_{\ell} \ge 1$. It is $\beta$-fractionally subadditive if $v(x) \le \beta \sum_{\ell} \alpha_{\ell}  v(y^{\ell})$ under the same condition.
\end{defn}

The above is the natural extension of the class of fractionally subadditive valuations that has been defined only for valuations defined on sets (i.e. special case of $\X_j\in \{0,1\}$). In the context of set valuations it has been shown that fractionally subadditive valuations is equivalent to the class of XOS valuations. We give here the natural generalization of XOS valuations in our setting and then we show that the analogous equivalence
theorem still holds, thereby extending the result of Feige \cite{Feige2006}.
We defer the proof to the Appendix.

\begin{defn}[XOS] A valuation is XOS if there exist a set $\El$ of additive valuations $v_{j}^{\ell}(x_j)$, such that: $v(x)=\max_{\ell\in \El} \sum_j v_{j}^{\ell}(x_j)$. It is $\beta$-XOS if:
\begin{align*}
\textstyle \max_{\ell\in \El} \sum_j v_{j}^{\ell}(x_j)\leq v(x) \leq \beta \max_{\ell\in \El} \sum_j v_{j}^{\ell}(x_j)
\end{align*}
\end{defn}

\begin{theorem}[XOS$\equiv$ Fractionally Subadditive]\label{thm:xos_equivalence}
A valuation is fractionally subadditive over the outcomes of different mechanisms if and only if it is XOS. Similarly, it is $\beta$-fractionally subadditive if and only if it is $\beta$-XOS.
\end{theorem}


\paragraph{Subclasses of Valuations}
To define generalizations of submodular and subadditive valuations, we will assume that each mechanism has a player-specific empty outcome $\bot_j\in \X_j$, 
which intuitively corresponds to: ''the mechanism is not existent for player $i$''. These outcomes don't affect the way the mechanism works (e.g. we don't impose that these outcomes be picked by the mechanism for some strategy profile) but it just serves as a reference point for the valuations of the bidders: we assume that $v(\bot_1,\ldots,\bot_m)=0$.  We will also use the notation $\x_S$ to denote the outcome vector that is $\x_j$ for all $j\in S$ and $\bot_j$ otherwise.
We start with 
the generalization of subadditivity of set valuations:
\begin{defn}A valuation is \textbf{set-subadditive} if and only if for any two sets $S_1, S_2 \subseteq [m]$ and any $\x\in \X$:
$$\textstyle{v\left(\x_{S_1}\right)+v\left(\x_{S_2}\right)\geq v\left(\x_{S_1\cup S_2}\right)}$$
\end{defn}
In addition we define the notion of set-submodularity which extends submodularity of set valuations as follows: the marginal benefit from receiving an allocation $\x_j$ at some mechanism $\Mj$ decreases as the set of mechanisms from which the agent has received a non-empty allocation becomes larger.
\begin{defn} A valuation is \textbf{set-submodular} if and only if,
for any $\x \in \X$ and for any two sets $S\subseteq T\subseteq [m]$:
$$\textstyle{v\left(x_{S+\{j\}}\right)-v\left(\x_S\right) \geq v\left(\x_{T+\{j\}}\right)-v\left(\x_{T}\right)}$$
\end{defn}

Last, we will make the intuitive assumption that if a player wins a non-empty allocation in more mechanisms then his valuation increases: a valuation is {\em set-monotone} if for any two sets $S\subseteq T$: $v\left(\x_{S}\right)\leq v\left(\x_{T}\right)$. 
We show that the relation between these classes of valuations mirrors the relations of the analogous classes for traditional valuations.
\begin{theorem}[Set-Submodular $\subseteq$ XOS]\label{thm:submod-xos} If a valuation is set - monotone and set-submodular then it is XOS.
\end{theorem}
\begin{theorem}[Set-Subadditive $\subseteq H_m$-XOS]\label{thm:subadd-xos} If a valuation is set-monotone and set-subadditive then it is $H_m$-XOS, where $H_m$ is the $m$-th harmonic number.
\end{theorem}

\paragraph{Restricted Classes of Valuations}
In some applications we want to consider restricted subclasses of valuations that are natural in the context. For such applications, we want to prove a strengthening of Theorem \ref{thm:xos_equivalence} where if the valuation $v(x)$ comes from some class, the component valuations $v_j^\ell: \X_j\rightarrow \R^+$ used in the equivalent XOS valuation also come from this class.

In the proof of Theorem \ref{thm:xos_equivalence} it is shown that any fractionally subadditive valuation can be expressed as an XOS valuation where each induced valuation $v_{j}^\ell$ is a single-minded valuation (i.e. $v_{j}^\ell(\x_j) = c$ if $\x_j=\hat{x}_j$ and $0$ otherwise).
However, single-minded valuations might not be natural in some applications. For example, if the outcome space of a component mechanism is 
ordered or is a lattice, then it is appropriate to consider only induced valuations that are monotone, or submodular on the lattice. We provide two such strengthenings of Theorem \ref{thm:xos_equivalence}.

%

\begin{theorem}\label{thm:monotone-xos}
If a valuation is monotone with respect to a coordinate-wise partial order $\succeq =(\succeq_j)_j$ and $\beta$-fractionally subadditive then it can be expressed as a  $\beta$-XOS valuation such that each induced valuation $v_{j}^\ell:\X_j\rightarrow \R^+$ is monotone with respect to $\succeq_j$.
\end{theorem}

If each poset $(\X_j,\succeq_j)$ forms a lattice then it is natural to consider valuations that have {\em diminishing marginal returns} over this lattice: i.e. for any $z\succeq y$ and $t \in \X$
\begin{equation*}
v(t\vee y)-v(y)\geq v(t\vee z)-v(z)
\end{equation*}
If the lattice is distributive and the valuation is monotone then
the above class of valuations is equivalent to the class of {\em submodular valuations
over the lattice} (proof in Appendix).
%
\begin{theorem}\label{thm:lattice-xos}
If a valuation is monotone and satisfies the diminishing marginal returns property with respect to a distributive product lattice $(\X,\succeq)$ then it can be expressed as an XOS valuation using 
valuations $v_{j}^\ell:\X_j\rightarrow \R^+$ that are \textit{capped marginal
valuations}: $$v_{j}^\ell(\x_j) = v(\x_j\wedge\hat{x}_j,\hat{x}_{-j})-v(\bot_j,\hat{x}_{-j})$$
(for some $\hat{x}\in \X$ associated with each $\ell$) and satisfy the diminishing marginal returns property with respect to  $(\X_j,\succeq_j)$.
\end{theorem}


\section{Smooth Mechanisms}
\label{SEC:SMOOTHNESS}

In this section we introduce the notion of a smooth mechanism for settings where agents have quasi-linear preferences. Our notion 
is similar to the smoothness of games of Roughgarden
\cite{Roughgarden2009}, but is tailored to the setting of mechanisms where participants have quasilinear preferences.

\begin{defn}[Smooth Mechanism]\label{def:smooth-mech}
A mechanism $\M$ is $(\lambda,\mu)$-smooth if for any valuation profile $v\in \times_i \Vi$ and for any action profile $a$ there exists a randomized action $\al_i^*(v,a_i)$ for each player $i$, s.t.:
\begin{equation}
\textstyle{\sum_i u_i^{v_i}(\al_i^*(v,a_i),a_{-i})\geq \lambda \opt(v)- \mu \sum_i P_i(a)}
\end{equation}
for some $\lambda,\mu\geq 0$. 
We denote by $u_i^{v_i}(\al)$ the expected utility of a player if $\al$ is a vector of randomized strategies.
\end{defn}
The definition of a smooth mechanism has a very natural interpretation as guaranteeing an approximate analog of market cleaning prices. Bikhchandani \cite{Bikhchandani1999} showed that pure Nash equilibria of a simultaneous first price auction have market clearing prices, and this implies that the outcome is efficient. Aggregate market clearing prices are guaranteed when each participant can modify her bid to claim her optimal bundle at the price paid for this bundle in the current solution.
$(1,1)$-smoothness in essence requires this property only in aggregate, but for any outcome of the mechanism, not only at equilibrium.  While $(\lambda,\mu)$-smoothness requires this only approximately, 
both in terms of the bundle claimed, as well as the price paid for it. In addition, unlike the pure equilibrium analysis, it requires the modified bid to be ignorant of the actions of the rest of the players.

We show that smooth mechanisms have low price of anarchy and that this result extends to all correlated equilibria (and hence learning outcomes) in the complete information setting and to all Bayes-Nash equilibria in the incomplete information setting without any change in the assumption.



\begin{theorem} \label{thm:smooth}
If a mechanism is $(\lambda,\mu)$-smooth and players have the possibility to withdraw from the mechanism then the expected social welfare 
at any Correlated Equilibrium of the game is at least $\frac{\lambda}{\max\{\mu,1\}}$ of the optimal social welfare.
\end{theorem}
\begin{proofsketch}
We prove the theorem for the case of a Pure Nash Equilibrium $a$. 
Since players have quasi-linear utilities we have: $v_i(X_i(a))=u_i(a)+P_i(a)$. Using that no player wants to deviate to $\al_i^*(v,a_i)$
we get:
\begin{align*}
\textstyle{\sum_i v_i(X_i(a))}\geq~& \textstyle{\sum_i u_i(\al_i^*(v,a_i),a_{-i})+ \sum_i P_i(a)} \\
\geq~ & \textstyle{\lambda \opt(v) +(1- \mu) \sum_i P_i(a)}
\end{align*}
The result follows if $\mu \le 1$. When $\mu >1$, to get the result, we note that $v_i(X_i(a))\ge P_i(a)$, as players have the possibility to withdraw from the mechanism and get 0 utility.
\end{proofsketch}


Our notion of smoothness of a mechanism differs from Roughgarden's notion of smoothness of games. To think of a mechanism as a game, we will consider the mechanism also as a player, with utility $\sum_i P_i(a)$ and no strategic decision to make. Our definition of a $(\lambda,\mu)$-smooth mechanism, is closely related to the game being $(\lambda,\mu-1)$-smooth in the sense of \cite{Roughgarden2009}, with two differences. We dropped the term $-(\mu-1)\sum_i u_i^{v_i}(a)$ on the right hand side, to make the definition more natural in the context of mechanisms. Note that this change makes the definitions incomparable, as with an arbitrary action profile $a$, the player utilities $u_i^{v_i}(a)$ can be
negative. Second, we allow the deviating strategy $a_i^*(v,a_i)$ to depend both on the valuation vector $v$ and the strategy of the deviating player $i$. This difference causes our Theorem \ref{thm:smooth} to only hold for correlated equilibria, and not coarse correlated equilibria. Allowing the deviating strategy to depend on $a_i$ makes it possible to prove a composability theorem for sequential mechanisms, where it is important to allow the deviating player to ``wait for the right moment'' to deviate. In games where the deviation required by smoothness does not depend on $a_i$, our results extend to coarse correlated equilibria. We focus on the version that allows this dependence so as to capture sequential composition. Simultaneous composition works well with either version of the definition.

\textbf{Incomplete Information Setting.}
Next we consider the case where the valuation of each player is drawn from a distribution $F_i$ over his valuation space $\Vi$. These distributions are independent and are common knowledge. A mechanism $\M=(\A,X,P)$ now defines a game of incomplete information.
The strategy of each player is a function $s_i: \Vi\rightarrow \Ai$. We will use
$s(v)=(s_i(v_i))_{i\in N}$ to denote the vector of actions given a valuation profile $v$ and $s_{-i}(v_{-i})=(s_j(v_j))_{j\neq i}$ to denote the vector of actions for all players except $i$.

The dominant solution concept in incomplete information games is the Bayes-Nash Equilibrium (BNE). A Bayes-Nash
Equilibrium is a strategy profile (possibly randomized) such that each player maximizes his expected utility conditional on his private information.
$$\forall a_i \in \Ai:\E_{v_{-i}|v_i}\left[u_i^{v_i}(s(v))\right]\geq \E_{v_{-i}|v_i}\left[u_i^{v_i}(a_i,s_{-i}(v_{-i})\right]$$
Given a strategy profile $s:\times_i \Vi\rightarrow \times_i \Ai$, 
our measure of efficiency will be the {\em expected social welfare over the valuations} of the players: $\E_{v}\left[SW^v(s(v))\right]$.
We will compare the efficiency of our solution concepts with respect to the
{\em expected optimal social welfare}: $\E_{v}\left[\opt(v)\right]$.

\textbf{Extension Theorem. }
The main result of this section is to show that if a mechanism is smooth according to definition
\ref{def:smooth-mech} then it achieves a good fraction of the expected optimal social welfare
at every Bayes-Nash equilibrium of the incomplete information game, irrespective of the
distributions of valuations. In the appendix we extend this result to general normal form games, strengthening the result of Roughgarden \cite{Roughgarden2012} and Syrgkanis \cite{Syrgkanis2012} where a strengthened notion of smoothness (universal smoothness) was used to establish efficiency results in the incomplete information setting. In addition, the previous definitions of smoothness in normal form games did not allow the deviating strategy to depend on the previous action of the deviating player and thereby wouldn't allow us to capture sequential games.

Note that the deviating strategy $\al_i^*(v,a_i)$ of player $i$ required by the smoothness property depends on the whole valuation profile $v$ and not only on the valuation of player $i$. As a result $\al_i^*(v,a_i)$ cannot be directly used as deviation for the player in the incomplete information game, as she is not aware of the valuations $v_{-i}$.
We use random sampling to handle the dependence on the values of other players, and a bluffing technique
to handle the dependence 
on the action of the deviating player.
\begin{theorem}\label{thm:extension-theorem}
If a mechanism $\M$ is $(\lambda,\mu)$-smooth and players have the possibility to withdraw, 
then for any set of independent distributions $F_i$, every mixed Bayes-Nash Equilibrium of the game induced by $\M$ has expected social welfare at least $\frac{\lambda}{\max\{\mu,1\}}$
of the expected optimal social welfare.
\end{theorem}
\begin{proof}
We will prove it for the case of a pure Bayes-Nash equilibrium $s(v)$ (the generalization to mixed equilibria is straightforward). Consider the following randomized deviation for each player $i$ that depends only on the information that he has which is his own value $v_i$ and the equilibrium strategies $s(\cdot)$: He random samples a valuation profile $w\sim \times_i F_i$. Then he plays 
$\al_i^*((v_i,w_{-i}), s_i(w_i))$, i.e., the player considers the equilibrium actions  $s(w)$, using the randomly sampled type (including the random sample of his own type), and deviates from this action profile using the action given by the smoothness property for his true type $v_i$, the random sample of the types of the others $w_{-i}$, and the equilibrium action $s_i(w_i)$ of his randomly sampled type $w_i$.  Using the action $s_i(w_i)$ as the base, corresponds to a bluffing technique that was introduced in \cite{Syrgkanis2012} in the context of sequential first price auctions, where player $i$ ``pretends'' that his valuation was $w_i$ until he deviates.

\noindent  Since this is not a profitable deviation for player $i$: 
\begin{align*}
\E_{v}\left[u_i^{v_i}(s(v))\right]\geq~& \E_{v,w} \left[u_i^{v_i}(\al_i^*((v_i,w_{-i}),s_i(w_i)),s_{-i}(v_{-i})\right]\\
=~& \E_{v,w}\left[u_i^{w_i}(\al_i^*((w_i,w_{-i}),s_i(v_i)),s_{-i}(v_{-i}))\right]\\
=~& \E_{v,w}\left[u_i^{w_i}(\al_i^*(w,s_i(v_i)),s_{-i}(v_{-i}))\right]
\end{align*}
Summing over players and using the smoothness property:
\begin{align}
\textstyle \E_v\left[\sum_i u_i^{v_i}(s(v))\right]\geq~& \textstyle \E_{v,w}\left[\sum_i u_i^{w_i}(\al_i^*(w,s_i(v_i)),s_{-i}(v_{-i}))\right]\nonumber\\
\geq~& \textstyle \E_{v,w}\left[\lambda \opt(w) - \mu\sum_i P_i(s(v))\right]\nonumber\\
=~& \textstyle \lambda \E_w\left[\opt(w)\right] - \mu \E_v\left[\sum_i P_i(s(v))\right]\nonumber
\end{align}
By quasi-linearity of utility and using the fact that players have the possibility to withdraw from the mechanism, we have the result.
\end{proof}


\section{Composition Theorems}
\label{SEC:COMPOSING}

{\bf Simultaneous Composition of Mechanisms.} For simultaneous composability of mechanisms we require that each mechanism is $(\lambda,\mu)$-smooth, and that the valuation is fractionally subadditive over outcomes of
mechanisms. To state the result more generally, recall that Theorem \ref{thm:xos_equivalence} implies that the valuation is also XOS.

\begin{theorem}[Simultaneous Composition]\label{thm:simultaneous}
Consider the 
simultaneous composition of $m$ mechanisms. Suppose that each mechanism $\Mj$ is $(\lambda,\mu)$-smooth when the mechanism restricted valuations of the players
come from a class 
$(\Vij)_{i\in [n]}$.
If the valuation $v_i:\Xsi\rightarrow \R^+$ of
each player across mechanisms is fractionally subadditive, and can be expressed as an XOS  valuation by component valuations $v_{ij}^{\ell}\in \Vij$ then the global mechanism is also $(\lambda,\mu)$-smooth.
\end{theorem}
\begin{proof}
Consider a valuation profile $v$ and an action profile $a$. Let $x^*$ be the optimal allocation for type profile $v$. Let $v_{ij}^*$ be the representative additive valuation for player $i$ for $x_i^*$ as implied by the definition of XOS valuations,
i.e. $v_i(x_i^*)=\sum_j v_{ij}^{*}(x_{ij}^*)$ and for all $\xsi\in \Xsi$: $v_i(x_i)\geq \sum_j v_{ij}^{*}(\xsij)$.

To prove the theorem we will show that there exists a deviation $\al_i^*=\al_i^*(v,a_i)$  of the global mechanism such that:
\begin{align*}
\textstyle\sum_i u_i^{v_i}(\al_i^*,a_{-i})\geq \lambda \sum_i v_i(x_i^*) - \mu\sum_i P_i(a)
\end{align*}
To define such a deviation we use the fact that each mechanism $\Mj$ is $(\lambda,\mu)$-smooth. Suppose that we
run mechanism $\Mj$ and each player has valuation $v_{ij}^*$ on $\Xsij$ and let $v_j^{*}$ be this
valuation profile. Since, by assumption those valuations fall in the valuation space $\Vij$ for which smoothness of $\Mj$ holds, for any action profile $\acj$ there exists a randomized action $\al_{ij}^*=\al_{ij}^*(v_j^{*},\acij)$ for each player, such that the sum of the utilities of the agents when each agent unilaterally deviates to it, is at least $\lambda\sum_i v_{ij}^{*}(x_{ij}^*) - \mu\sum_i P_i^j(\acj)$.

For the global mechanism, we consider a randomized deviation $\al_i^*=\al_i^*(v,a_i)$ of player $i$ that consists of independent randomized deviations $\al_{ij}^*=\al_{ij}^*(v_j^{*},\acij)$ for each mechanism $j$
as described in the previous paragraph. For each action $a_i^*$ in the support of $\al_i^*$ we denote with $X_i(a_i^*,a_{-i})$ the outcome vector in that action profile.  By the properties of the representative additive valuation, we have that
$v_i(X_i(a_i^*,a_{-i}))\geq \sum_j v_{ij}^{*}(X_i^j(a_{ij}^*,a_{-i}^j))$.
Thus the expected utility of player $i$ from the deviation will be at least:
\begin{align*}
 u_i^{v_i}(\al_i^*,a_{-i})
\geq~&
 \textstyle{\E_{\al_i^*}\left[\sum_{j}v_{ij}^{*}\left(X_i^j(\al_{ij}^*,a_{-i}^j)\right)-P_i^j(\al_{ij}^*,a_{-i}^j)\right]}
\end{align*}
Now adding over all players $i$ we have:
\begin{align*}
\textstyle{\sum_i u_i^{v_i}(\al_i^*,a_{-i})}\geq\textstyle{
\sum_{j,i} \E_{\al_{ij}^*}\left[v_{ij}^{*}(X_i^j(\al_{ij}^*,a_{-i}^j))-P_i^j(\al_{ij}^*,a_{-i}^j)\right]}
\end{align*}
The key argument is that
\begin{equation*}
\textstyle{\sum_{i}\E_{\al_{ij}^*}
\left[v_{ij}^{*}(X_i^j(\al_{ij}^*,a_{-i}^j))-P_i^j(\al_{ij}^*,a_{-i}^j)\right]}
\end{equation*}
is the sum of the expected utilities where starting from strategy profile $a^j$ each player $i$ unilaterally deviates to a randomized bid $\al_{ij}^*=\al_{ij}^*(v_j^{*},a_i^j)$ in mechanism $\Mj$ and when each player $i$ has valuation $v_{ij}^{*}$ for the different outcomes $\xsij$ of mechanism $\Mj$.
By smoothness of each mechanism $\Mj$:
\begin{align*}
\textstyle \sum_i u_i^{v_i}(\al_i^*,a_{-i})\geq~& \textstyle \sum_{j}\left(\lambda \sum_i v_{ij}^{*}(x_{ij}^*)-\mu \sum_i P_i^j(a^j)\right)\\	
=~&\textstyle \lambda \sum_i \sum_{j}v_{ij}^{*}(x_{ij}^*) - \mu \sum_i \sum_j P_i^j(a^j)\\
=~&\textstyle \lambda \sum_i v_i(x_i^*)-\mu \sum_i P_i(a)
\end{align*}
where we used that by the definition of the representative additive valuation $v_i(x_i^*)=\sum_{j\in [m]}v_{ij}^{*}(x_{ij}^*)$.
\end{proof}

In Sections \ref{SEC:SINGLE-ITEM} and \ref{SEC:APPLICATIONS} we will give a number of applications of this result.
Note that if the classes $\Vij$ contain single-minded valuations (e.g. in the case of combinatorial auctions), then our composability theorem holds for any fractionally subadditive valuation. For  classes of valuations that do not contain single-minded valuations (such as ad-auctions), we can apply the 
theorem using results from Section \ref{SEC:VALUATIONS} on
monotone and lattice-submodular valuations.

{\bf Sequential Composition of Mechanisms.}
In many scenarios, mechanisms might not all take place simultaneously. Sequentiality however can lead to inefficiencies as was shown by recent works on sequential auctions \cite{PaesLeme2012, Syrgkanis2012}. Here, we show that the positive results of \cite{PaesLeme2012, Syrgkanis2012} on unit-demand sequential first price auctions are a special case of a more general property of smooth mechanisms.
For the sequential composition of $m$ mechanisms we prove that if each mechanism $\Mj$ is $(\lambda,\mu)$-smooth, then  the resulting mechanism is $(\lambda,\mu+1)$-smooth (for the normal form representation of the extensive form of game) if an agents valuation is the best of her valuation over the different mechanisms: 
$v_i(\xsi)=\max_{j\in [m]} v_i^j(\xsij)$.

An interesting aspect of the sequential composition is that the strategy of a player is no longer just an action $\acij\in \Aij$ for each mechanism but rather a whole contingency plan of what action she will submit to mechanism $\Mj$ conditional on any observed history of play. Our result doesn't depend on what part of the history is observed by the players, whether players just observe their own allocation, or all allocations, or also all prices, or bids. We don't even need that all players observe the same things.
However, 
we assume that the information structure 
is common knowledge.

\begin{theorem}[Sequential Composition]\label{thm:sequentially}
Consider the sequential composition of $m$, $(\lambda,\mu)$-smooth, mechanisms
defined on valuation spaces $\Vij$.
If each valuation $v_i:\X_i\rightarrow \R^+$ of
is of the form $v_i(\xsi)=\max_{j\in [m]} v_i^j(\xsij)$, with  $v_i^j\in \Vij$, then the global mechanism is $(\lambda,\mu+1)$-smooth, independent of the information released to players during the sequential rounds.
\end{theorem}

We can combine  these two 
theorems to prove efficiency guarantees
when mechanisms are run in a sequence of rounds and at each round several mechanisms are run simultaneously.

%

\section{An Application: Item Auctions}
\label{SEC:SINGLE-ITEM}

In this section we present a simple, yet rich, 
application of our framework to the case where each component mechanism is a single-item auction.
We consider the three main single-item auctions: first-price, all-pay and second-price.

\textbf{First Price Auction.} The first price auction is a $(1-1/e,1)$-smooth mechanism. To see the smoothness note that under any valuation profile $v$ (note that we only need to argue about the full information setting), the highest value player with value $v_{h}$ can deviate to submitting a randomized bid $b_{h}'$ drawn from a distribution with density function $f(x)=\frac{1}{v_{h}-x}$ and support $[0,(1-1/e)v_{h}]$, while all non-highest value players should just deviate to bidding $0$.  No matter what the rest of the players are bidding, the utility of the highest bidder from the deviation is:
\begin{align*}
u_{h}^{v_{h}}(b_{h}',b_{-h}) \geq~& \textstyle{\int_{\max_{i\neq h}b_i}^{\left(1-\frac{1}{e}\right)v_h} (v_h-x)f(x)dx}\\
\geq~& \textstyle\left(1-\frac{1}{e}\right)v_{h}-\max_{i}b_i
\end{align*}
Theorem \ref{thm:simultaneous} now implies that if we run $m$ simultaneous first price auctions and bidders have fractionally subadditive valuations then
any Correlated Equilibrium in the full information setting and any mixed Bayes-Nash Equilibrium in the
incomplete information setting has social welfare at least $(1-\frac{1}{e})$ of the optimal. A looser result
of $1/4$ for this setting and only for mixed Nash and Bayes-Nash appeared in \cite{Hassidim2011}. The tighter
result of $(1-\frac{1}{e})$ appeared in \cite{Syrgkanis2012}. Theorem \ref{thm:sequentially} implies that if we run $m$ first price auctions sequentially and bidders have unit-demand valuations then any Correlated Equilibrium in the full information setting and any Bayes-Nash Equilibrium in the incomplete information setting has social welfare at least $\frac{1}{2}(1-\frac{1}{e})$ of the optimal. The latter result was given in a sequence of two papers
\cite{PaesLeme2012,Syrgkanis2012}.

\textbf{All-Pay Auction.} The all-pay auction is a $(1/2,1)$-smooth mechanism. The smoothness proof is similar to the first price auction with the only alteration that we make the highest value player submit a bid drawn uniformly at random from $[0,v_{h}]$.
The utility from such a deviation is:
\begin{align*}
 u_{h}^{v_{h}}(b_{h}',b_{-h}) \geq~& \textstyle{\int_{\max_{i\neq h}b_i}^{v_{h}} v_{h}f(x)dx-\E[b_{h}']}\\
\geq~& \textstyle{\frac{1}{2}v_{h}-\max_{i}b_i \geq \frac{1}{2}v_{h}-\sum_i b_i}
\end{align*}
Therefore we get an efficiency guarantee of $1/2$ for the simultaneous composition of $m$ all-pay auctions
and an efficiency guarantee of $1/4$ for the sequential composition both in the Bayesian setting and in learning outcomes. Simultaneous and sequential
all-pay auctions have not been studied in the literature and could prove useful in capturing simultaneous
or sequential all-pay contests, which is a natural model for several online crowd-sourcing environments.

\textbf{Second-Price Auction.} 
The second price auction is not a smooth mechanism. In fact, the
second price auction
is not as robust as the previous auctions.
Second price auctions have arbitrary bad equilibria when players bid above their value,
Goeree \cite{Goeree2003} shows that signaling is bound to arise in a second price auction when bidders are strategising about future opportunities, and Paes Leme et al \cite{PaesLeme2012} show an example with unbounded inefficiency when running second price auctions sequentially and bidders are unit-demand. The main difference of the second price auction and the previous two auctions is that it makes very loose connection between the bid a player needs to make to win and the price that was previously paid to the auctioneer.
Several papers \cite{Christodoulou2008,Bhawalkar2011,Markakis2012,Caragiannis2012} have used an assumption that players will not bid above their valuations to give good efficiency guarantees for second-price type of auctions. 
Next, we extend our results to mechanisms that require such no-overbidding assumptions.

\section{Weak Smoothness}\label{SEC:WEAK}

In this section we give a generalization of our framework  to capture mechanisms that produce high efficiency
under a no-overbidding refinement.
First, we give a definition of no-overbidding that generalizes
the no-overbidding assumptions used in the literature \cite{Christodoulou2008,Bhawalkar2011,Caragiannis2012}. In a single-item second-price auction the bid of a player is his maximum willingness to pay
when he wins. The following defines maximum willingness to pay in the general mechanism
design setting.
\begin{defn}[Willingness-to-pay]\label{def:willingness-to-pay} Given a mechanism $(\A,X,P)$ a player's
maximum willingness-to-pay for an allocation $x_i$ when using strategy $a_i$ is defined
as the maximum he could ever pay conditional on allocation $x_i$:
\begin{equation}
 B_i(a_i,x_i) = \textstyle{\max_{a_{-i}:~ X_i(a)=x_i} P_i(a)}
\end{equation}
\end{defn}
\begin{defn}[Weakly Smooth Mechanism]\label{def:weak-smooth} A mechanism is weakly $(\lambda,\mu_1,\mu_2)$-smooth for $\lambda,\mu_1,\mu_2\geq 0$, if for any type profile $v\in \times_i \Vi$ and for any action profile $a$ there exists a randomized action $\al_i^*(v,a_i)$ for each player $i$, s.t.:
\begin{multline*}
\textstyle{\sum_i u_i^{v_i}(\al_i^*(v,a_i),a_{-i})}\geq \textstyle{\lambda \opt(v)- \mu_1 \sum_i P_i(a)}\\
-\textstyle{\mu_2\sum_i B_i(a_i,X_i(a))}
\end{multline*}
\end{defn}
\begin{defn}[No-overbidding] A randomized strategy profile $\al$ satisfies the no-overbidding assumption if:
\begin{equation}
\textstyle{\E_{\al}[B_i(\al_i,X_i(\al))]\leq \E_{\al}[v_i(X_i(\al))]}
\end{equation}
i.e., at this strategy profile no player is bidding in a way that she could potentially pay more than her value subject to her expected allocation remaining the same.
\end{defn}
\begin{theorem}\label{thm:weak-efficiency} If a mechanism is weakly $(\lambda,\mu_1,\mu_2)$-smooth then any Correlated Equilibrium in the full information setting and any mixed Bayes-Nash
Equilibrium in the Bayesian setting that satisfies the no-overbidding assumption achieves efficiency at least $\frac{\lambda}{\mu_2+\max\{\mu_1,1\}}$
of the expected optimal.
\end{theorem}
In the Appendix, we show, analogously to the results in Section \ref{SEC:COMPOSING}, that the simultaneous composition of weakly $(\lambda,\mu_1,\mu_2)$-smooth mechanisms is weakly $(\lambda,\mu_1,\mu_2)$-smooth and the sequential composition is weakly $(\lambda,\mu_1+1,\mu_2)$-smooth.

\textbf{Remark 1.} In contrast to the smoothness used in \cite{Roughgarden2012} our definition of smoothness allows us to prove efficiency
under the weaker assumption of no-overbidding in expectation, rather than point-wise
no-overbidding. The main difference is that we incorporate the willingness-to-pay
inside the smoothness definition, while previous smoothness approaches would relate to value 
directly. The latter approach would require to use point-wise no-overbidding to relate bids to welfare in second-price auctions.

{\bf Remark 2.} We use the non-overbidding assumption as an equilibrium refinement rather than as a strategy-space restriction. Several papers in the literature have used non-overbidding as a strategy space restriction (rather than as an equilibrium refinement). The two uses are equivalent in settings where the restricted strategy space always contains best-responses. Note that while overbidding is a dominated strategy in a single item auction, global no-overbidding is not dominated when running second price auctions simultaneously or sequentially. Overbidding equilibria that survive elimination of dominated strategies and that have non-constant inefficiency have been given both for the case of sequential \cite{PaesLeme2012} and simultaneous \cite{Feldman2012} second price auctions, even in the simplest scenario when bidders are unit-demand. Restricting the strategy space to non-overbidding strategies, could potentially create artificial equilibria that were not equilibria of the original game, since this restricted strategy space does not always contain best-responses (see \cite{Feldman2012} for an example). On the other hand, the refined set of non-overbidding equilibria might be empty.
Some of our results carry over to the strategy-space restriction version and a
detailed exposition is deferred to the full version.

\section{Budget Constraints}\label{SEC:BUDGET}
An important class of non-quasilinear preferences
is when players have hard budget constraints on the payments they make.
Studying the effect of budgets on efficiency has received great attention in recent algorithmic game theory literature
\cite{Dobzinski2008, Fiat2011, Goel2012, Duetting2012} mostly in the realm of truthful
mechanism design and assuming that the budgets are common knowledge.
Little is known about the effect of budgets in the case of non-truthful mechanisms.
For instance, only recently Huang et al. \cite{Huang2012} analyzed efficiency in a two-player sequential first price auction game with budget constraints in the complete information setting.

Most of the literature 
has focused on producing
pareto-optimal outcomes, i.e. a pair of allocation and prices
such that there is no other pair that respects feasibility and budget constraints
and such that all players receive strictly higher utility and the auctioneer receives strictly higher
revenue.

We 
study an orthogonal benchmark, which we call {\em Effective Welfare}, obtained by capping a player's value by his budget:
\begin{equation}
\textstyle{EW(x) = \sum_i \min\{v_i(\xsi),B_i\}}
\end{equation}
We compare the social welfare resulting in our mechanism to the maximum possible effective welfare. This benchmark reflects that we cannot expect players with low budgets to be effective at maximizing their own value.

We show that a lot of our results carry over to the effective welfare benchmark, by introducing a strengthening of the smoothness property of mechanisms; a strengthening that is is satisfied by almost all the applications we consider. We focus on smooth mechanisms, but all the results in this section extend to weak smoothness assuming no-overbidding.
\begin{defn}[Conservatively Smooth Mechanism]
A mechanism is conservatively $(\lambda,\mu)$-smooth if it is $(\lambda,\mu)$-smooth in the quasilinear utility setting and the actions in the support of the smoothness deviations satisfy:
\begin{equation}
 \textstyle{\max_{a_{-i}} P_i(a_i^*(v,a_i),a_{-i})\leq \max_{\xsi\in \Xsi} v_i(\xsi)}
\end{equation}
\end{defn}

The next theorem shows that the expected social welfare at Correlated Equilibria and at Bayes-Nash equilibria of conservatively smooth mechanisms is a good fraction of the optimal effective welfare.
Note that in the incomplete information setting, the private information of a player is 
his valuation and his budget. We will denote the valuation and budget pair as the type $t_i=(v_i,B_i)$ of player $i$ and we will assume that it 
is distributed independently according to some distribution $F_i$ on $\Vi\times \R^+$. Note that we allow 
the budget of a player to be correlated with his valuation.

\begin{theorem}\label{thm:conservative-efficiency}
 If a mechanism is conservatively $(\lambda,\mu)$-smooth and its valuation space is closed
under capping, then the social welfare at any correlated equilibrium and at any Bayes-Nash equilibrium
is at least $\frac{\lambda}{\max\{1,\mu\}}$ of the expected maximum effective welfare.
\end{theorem}

Last we show that efficiency guarantees for budget-constraint bidders are composable under the
conservative smoothness property for simultaneous composition. Unfortunately, sequential composition doesn't carry over. In sequential mechanisms  a good deviation may require that the player waits and plays according to equilibrium until his optimal mechanism arrives. While ''waiting'' he might exhaust his budget.

\begin{theorem}\label{thm:budget-composition}
Consider the simultaneous composition of $m$ conservatively $(\lambda,\mu)$-smooth mechanisms defined on valuation spaces $\Vij$ that are closed under capping. If players have XOS valuations and can be expressed by valuations $v_{ij}^\ell\in \Vij$ then the social welfare at any correlated equilibrium and at any Bayes-Nash equilibrium of the global mechanism is at least $\frac{\lambda}{\max\{1,\mu\}}$ of the expected maximum effective welfare.
\end{theorem}
The composability result is proved in a sequence of two lemmas: 
first we prove that conservative smoothness of a mechanism composes under XOS valuations and second we show that if the valuation space of each component mechanism is closed under capping then the corresponding valuation space of the composition mechanism
is also closed under capping. The latter is shown by proving a structural property of XOS valuations: a valuation produced by capping an XOS valuation is also XOS and can be described by component valuations that are cappings of the component valuations of 
the XOS representation of the initial valuation. Using these two lemmas we can invoke Theorem \ref{thm:conservative-efficiency} to get efficiency guarantees for budget constrained bidders in the global mechanism.


\section{Applications}
\label{SEC:APPLICATIONS}
In this section we give several applications of our framework. Some 
are new smoothness proofs implying new bounds on efficiency, 
others are reinterpretations of existing literature as smoothness proofs.
In each case adding budget constraints gives new results on efficiency of mechanisms, and our results show that the efficiency is preserved by composition of mechanisms.
The efficiency guarantees 
hold for correlated equilibria in the full information setting and for  mixed Bayes-Nash equilibria in the incomplete information setting. Our guarantees are with respect to the optimal effective welfare when the players have budget constraints.

\textbf{Single Item Auctions.} Extending the results of Section \ref{SEC:SINGLE-ITEM} we show that
the first price single item auction
is conservatively $\left(1-\frac{1}{e},1\right)$-smooth, the all-pay auction is conservatively $\left(\frac{1}{2},1\right)$-smooth and the second price
auction is weakly and conservatively $(1,0,1)$-smooth.
We also give a smoothness proof for the hybrid auction
in which the winner pays a convex combination of her own bid and the second highest bid. Our framework
implies that running $m$ simultaneous first price auctions and bidders have fractionally
subadditive valuations and budget constraints achieves efficiency at least $1-\frac{1}{e}$ of the optimal effective welfare. All-pay auctions achieve a guarantee of $\frac{1}{2}$. Second price auctions achieve a guarantee of $\frac{1}{2}$ under the no-overbidding assumption. For sequential auctions with unit-demand bidders and no budget constraints the first price, all-pay and second price auctions give guarantees of $\frac{1}{2}(1-\frac{1}{e})$, $\frac{1}{4}$ and $\frac{1}{2}$ respectively.

\textbf{Greedy Direct Auctions.} Lucier and Borodin 
\cite{Lucier2010} considers combinatorial auctions,
whose allocation function is based on a greedy $c$-approximation algorithm.
When a first price payment is used, they show that such a greedy auction has a $c+O(\log(c))$ efficiency guarantee.
We improve this bound, by showing that this mechanism is conservatively $(1-e^{-1/c},c)$-smooth implying an efficiency guarantee of at least $\frac{1}{c+0.58}$.
This bound extends to the simultaneous composition of such mechanisms when bidders have fractionally subadditive valuations across auctions
and budget constraints. For example, when each auctions sells only a small number of items, greedy algorithms can do quite well (giving a $\sqrt{k}$-approximation for arbitrary valuations, if each auction sells at most $k$ items). 
Observe, that fractionally subadditive valuations across auctions allow for complements within the items of a single greedy auction, hence is more general than just assuming that players have fractionally subadditive valuations over the whole universe of items.
In the appendix, we show that the above analysis is a special case of a more general class of {\em direct auctions. 
}
%

\textbf{Position Auctions.} We analyze position auctions for more general valuation spaces than
what has been typically considered \cite{Edelman2007, Caragiannis2012}. We use the model of Abrams et al \cite{Abrams2007}, where each player $i$ has an arbitrary valuation $v_{ij}$ for appearing at position $j$, that is monotone in the position. Most of the literature in position auctions has considered valuations of the
form $v_{ij}=a_j \gamma_i v_i$, i.e. players have only value per click $v_i$ and their click-through-rate is dependent in a separable way on their quality and on the position. The more general class of valuations can capture settings
where players have value both for click and for the impression itself, and settings where the click-through-rates are not separable.
We show that the following very simple first price analog of the auction of \cite{Abrams2007} is conservatively $(\frac{1}{2},1)$-smooth: solicit bids from the players, allocate positions in order of bids and charge each player his bid.
The implied guarantee of $\frac{1}{2}$
holds for simultaneous composition when players have monotone fractionally subadditive valuations and budget constraints. Such valuations capture, for instance, settings where bidders have value $v_i$ only for the first $k$ clicks,
or settings where the marginal value per-click of a player 
decreases with the number of clicks he gets. In addition a bound of $\frac{1}{4}$ is implied for the sequential composition when bidders value is the maximum value among all impressions he got. In contrast, \cite{Abrams2007} consider the second price analog of this auction, and show that it always has an efficient Nash equilibrium, but do not consider the price of anarchy. We show that the second price version is conservatively weakly $(\frac{1}{2},0,1)$-smooth, implying an efficiency guarantee of $\frac{1}{4}$ for simultaneous and sequential composition of such auctions under the no-overbidding assumption. In the appendix we also consider other variations of the well-studied GFP and GSP mechanisms for the case when players have only values per click.

\textbf{Bandwidth Allocation Mechanisms.} We consider the setting studied by Johari and Tsitsiklis 
\cite{Johari2004} where a set of players want to share a resource: an edge with bandwidth $C$. Each player has a concave
valuation $v_i(x_i)$ for getting $x_i$ units of bandwidth. The mechanism studied in \cite{Johari2004} is the following:
solicit bids $b_i$, allocate to each player bandwidth proportional to his bid $x_i=\frac{b_i}{\sum_j b_j}$,
charge each player $b_i$. We show that this mechanism is conservatively $(2-\sqrt{3},1)$-smooth, implying 
an efficiency guarantee of approximately $1/4$ for 
correlated equilibria and Bayes-Nash equilibria. 
The same efficiency guarantee extends to the case when we run such mechanisms simultaneously and players
have budget constraints and monotone, lattice-submodular valuations on the lattice defined on $\R^m$ by the coordinate-wise ordering. If the valuations are twice differentiable, being monotone and lattice-submodular translates to: 
every partial derivative is non-negative and every cross-derivative is non-positive.

\textbf{Multi-Unit Auctions.} For the setting of multi-unit auctions where players have concave utilities in the amount of units they get, we give two 
smooth mechanisms. Recently, Markakis et al. \cite{Markakis2012} studied the following greedy mechanism: solicit marginal bids from
the agents $b_{ij}$ ($b_{ij}$ is the declared marginal value of agent $i$ for the $j$-th unit), at
each iteration pick the maximum 
marginal bid conditional on the current allocation and
allocate the extra unit, until all units are allocated. Markakis et al. \cite{Markakis2012} studied
a uniform-price auction where each player is charged the lowest unallocated marginal bid, for
every unit she got and showed a $O(\log(m))$ approximation for the case of mixed Bayes-Nash equilibria under a no-overbidding assumption. Here, we show that a first price version of the above mechanism where each player is charged his declared marginal bids for the items he acquired is conservatively
$\left(\frac{1}{2}\left(1-\frac{1}{e}\right),1\right)$-smooth, while the uniform price version of \cite{Markakis2012} is
weakly $(\frac{1}{2}\left(1-\frac{1}{e}\right),0,1)$-smooth, when the willingness-to-pay of an agent is the
sum of his $k_i$ highest marginal bids when allocated $k_i$ units. Therefore our smooth analysis improves the $O(\log(m))$ bound of \cite{Markakis2012} to a constant $\frac{1}{4}\left(1-\frac{1}{e}\right)$ and
to $\frac{1}{2}\left(1-\frac{1}{e}\right)$ when a first price payment rule is used. In addition, the above bounds
carry over to simultaneous composition under budget constraints and when bidders have monotone and lattice-submodular
valuations on the lattice $\N^{m}$.
We show that a simpler uniform-price auction 
is also weakly
$(\frac{1}{2}\left(1-\frac{1}{e}\right),0,1)$-smooth: solicit 
a quantity $q_i$ and a per-unit bid $b_i$,
consider bids in decreasing order and allocate greedily until all units are sold. The per-unit price for everyone is the last unallocated bid.

\vspace{-.07in}\bibliographystyle{abbrv}
\bibliography{bayesian_smoothness}
\newpage
\appendix
\section{Applications}
Our work provides some new results in the context of efficiency of non-truthful mechanisms and unifies previous work. For each application we will show how smooth each mechanism is. Then we will highlight some of the implications that our framework implies. For conciseness we will not list all the implications of our framework for each mechanism, but one can apply all our general theorems for each of the applications. In our efficiency theorems for conciseness we will refer to a correlated equilibrium in the full information setting as CE and to a mixed Bayes-Nash equilibrium in the incomplete information setting as BNE. When we refer to expected welfare then this would be over the randomness of the action profiles in the complete information setting and over the randomness of the valuations, budgets and action profiles in the incomplete information setting. When we refer to settings with budget constraints our bounds are with respect to the optimal effective welfare.

\subsection{Single Item Auctions} In this section we revisit the three main single-item auctions discussed in Section \ref{SEC:SINGLE-ITEM} as well as the hybrid auction where the winner pays a mixture of his bid and the second highest bid and give a complete list of our results.

\textbf{First Price Auction.} As explained in Section \ref{SEC:SINGLE-ITEM} a first
price auction is $(1-\frac{1}{e},1)$-smooth since for any valuation profile the highest value player with value $v_{\max}$ can deviate to submitting a randomized bid $b_{\max}'$ drawn from a distribution with density function $f(x)=\frac{1}{v_{\max}-x}$ and support $[0,(1-1/e)v_{\max}]$. Here we observe that the above deviation
also implies conservative smoothness.
\begin{corollary}
 The first price single-item auction is conservatively $(1-\frac{1}{e},1)$-smooth.
\end{corollary}

\begin{corollary}[Simultaneous with Budgets]
 If we run $m$ simultaneous first price auctions and bidders have budgets and fractionally
subadditive valuations then every CE and BNE achieves at least $\frac{e-1}{e}\approx 0.63$ of the expected optimal effective welfare.
\end{corollary}

\begin{corollary}[Sequential] If we run $m$ sequential first-price auctions with unit-demand bidders then
every CE and BNE achieves $\frac{1}{2}\frac{e-1}{e}\approx 0.32$ of the expected optimal social welfare.
\end{corollary}

\textbf{All-Pay Auction.} The all-pay auction is $(1/2,1)$-smooth since the highest value player submit a bid drawn uniformly at random from $[0,v_{\max}]$ as shown in section \ref{SEC:SINGLE-ITEM}. Observe again that this deviation also implies conservative smoothness. Hence:

\begin{corollary}[Simultaneous with Budgets]
 If we run $m$ simultaneous all-pay auctions and bidders have budgets and fractionally
subadditive valuations then the expected effective welfare at every CE
in the complete information case and at every BNE in the incomplete information case is at least $1/2$ of the expected optimal effective welfare.
\end{corollary}

\begin{corollary}[Sequential] If we run $m$ sequential all-pay auctions with unit-demand bidders then
every CE and BNE achieves $1/4$ of the expected optimal social welfare.
\end{corollary}

\textbf{Second Price Auction.} In a second-price auction the winning bidder pays the second highest bid.
Here we show that the  second price auction is weakly $(1,0,1)$-smooth. Observe that in a hybrid auction the willingness to pay of a winning bidder is exactly his bid. This can be easily shown since the highest value player can switch to bidding his true value in which case his utility is at least $v_{\max}-b$ where $b$ was the
highest bid in the previous strategy profile and hence the willingness-to-pay of the winning bidder
in the previous strategy profile.
\begin{lemma}
 The second price auction is weakly $(1,0,1)$-smooth.
\end{lemma}

\begin{corollary}[Simultaneous with Budgets]
 If we run $m$ simultaneous second price auction and bidders have budgets and fractionally subadditive valuations then any CE and BNE that satisfies the weak no-overbidding assumption globally, achieves at least $1/2$ of the optimal social welfare.
\end{corollary}

\begin{corollary}[Sequential] If we run $m$ sequential second price auctions with unit-demand bidders then
every CE and BNE that satisfies the no-overbidding assumption achieves $1/2$ of the expected optimal social welfare.
\end{corollary}

\textbf{Hybrid Auction.} In the $\gamma$-hybrid auction the winner pays his bid with probability $\gamma$ and the second highest bid with probability $(1-\gamma)$.

\begin{lemma}
 The $\gamma$-hybrid auction is weakly $$(\gamma(1-\frac{1}{e})+(1-\gamma)^2,1,(1-\gamma)^2)$$ smooth.
\end{lemma}
\begin{proof}
 Consider a valuation profile $v$ and a bid profile $b$. Let $b_{\max}$ the highest bid and $v_{\max}$ the highest value. The non highest value bidders deviation is bidding $0$. The highest value bidder's deviation is bidding as follows: With probability $\gamma$ he submits a bid according to distribution with density $f(t)=\frac{1}{v_{\max}-t}$ and support $[0,(1-\frac{1}{e})v_{\max}]$. With
 probability $(1-\gamma)$ he submits his true value.

 In the first case the utility of the bidder is at least:
 \begin{multline*}
  \int_{b_{\max}}^{(1-\frac{1}{e})v_{\max}} (v_{\max}-\gamma t - (1-\gamma)b_{\max})f(t) dt \geq\\
  \int_{b_{\max}}^{(1-\frac{1}{e})v_{\max}} (v_{\max}-t)f(t) dt=
  \left(1-\frac{1}{e}\right)v_{\max} - b_{\max}
 \end{multline*}

 In the case when he submits his true value then when $b_{\max}< v_{\max}$ he wins and gets utility
 $$v_{\max}-\gamma v_{\max} - (1-\gamma)b_{\max} = (1-\gamma)(v_{\max}-b_{\max})$$ When
 $b_{\max}\geq v_{\max}$ he either loses or ties and in any case gets non-negative utility
 and thereby utility at least $(1-\gamma)(v_{\max}-b_{\max})$.

 Thus overall the expected utility from the deviation is at least:
 \begin{multline*}
  \gamma \left(1-\frac{1}{e}\right)v_{\max} - \gamma b_{\max} +(1-\gamma)^2(v_{\max}-b_{\max})
 \end{multline*}
 The lemma follows by just observing that the payment under bid profile $b$ is at least $\gamma b_{\max}$
\end{proof}

\begin{corollary}[Simultaneous with Budgets]
 If we run $m$ simultaneous $\gamma$-hybrid auctions and bidders have budgets and fractionally
subadditive valuations then every CE and BNE that satisfies the no-overbidding assumption achieves $\frac{\gamma(1-\frac{1}{e})+(1-\gamma)^2}{1+(1-\gamma)^2}$ of the expected optimal social welfare.
\end{corollary}

\begin{corollary}[Sequential] If we run $m$ sequential $\gamma$-hybrid auctions with unit-demand bidders then
every CE and BNE achieves $\frac{\gamma(1-\frac{1}{e})+(1-\gamma)^2}{2+(1-\gamma)^2}$ of the expected optimal social welfare.
\end{corollary}

\begin{figure}
\centering
\epsfig{file=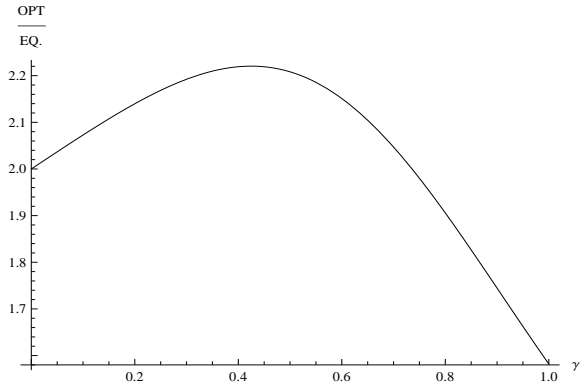, height=2in, width=3in}
\caption{Efficiency bound for the hybrid auction.}
\end{figure}

\subsection{Direct Auctions.}
In this section we will focus on mechanisms where players directly report their values. These mechanisms show an interesting use of threshold bids (critical also in truthful mechanism design) in \emph{smooth mechanism design}. In a direct mechanism the player declares a value $\tilde{v}_i(x_i)$ for any allocation outcome $x_i$. Thus the action space $\A_i$ of each player defined by the mechanism is equal to the set of valuations $\V_i$.

We will focus on direct mechanisms that are also ex-post individually rational. We apply the individual rationality constraint to non-truthful mechanisms in the sense of requiring that if reporting valuations truthfully the resulting utility of any agent is non-negative point-wise over every randomness of the mechanism.
Note that this  doesn't imply that the player do reports truthfully. It just places the the restriction on the allocation function $X:\A\rightarrow \X$ and the payment function $P:\A\rightarrow \R^n_+$ of the mechanism that $P_i(\tilde{v})\leq \tilde{v}_i(X_i(\tilde{v}))$.

To motivate the connection to truthful mechanism design, we first describe a single-dimensional service-based mechanism design setting. Consider a setting where the allocation space is just a feasibility set of players that can be served. In other words the outcome space from the perspective of each player is binary $\X_i=\{0,1\}$ and the
outcome space of the mechanism is some feasible subset of the product space. In addition, suppose that the value of a player was just a number $v_i$ for being served. The utility of an agent would then be of the form $u_i(x,p)=v_i x_i-p_i$, where $x_i$ is either $0$ or $1$. It is a well known result that
an efficient truthful direct mechanism needs to charge player $i$, the minimum value he needs to have to still be allocated: $P_i(\tilde{v})=\tau_i(\tilde{v}_{-i})=\inf\{z: X_i(z,\tilde{v}_{-i})=1\}$ and $0$ if he is not allocated.

Is there a similar characterization for the more general setting? For more general settings the
standard mechanism that is efficient and guarantees non-negative prices and individual rationality
is the VCG mechanism. Unfortunately, the VCG mechanism doesn't have a similar simple "threshold bid"
interpretation. Despite this fact, one could still define threshold bids in the more general quasi-linear setting as follows:

\begin{defn}Given a direct mechanism $M$ and a bid profile $\tilde{v}$, we say that the threshold bid $\tau_i(x_i,\tilde{v}_{-i})$ of player $i$ for allocation $x_i$ is the minimal value that player $i$ has to single mindedly declare for allocation $x_i$ such that he is allocated $x_i$ whenever $\tilde{v}_i(x_i)\geq \tau_i(x_i,\tilde{v}_{-i})$.
\end{defn}

To make a mechanism truthful in a single parameter setting remember that one had to strongly tie together
the threshold bids of the players with their actual payments. In what follows we show that even in smooth mechanism design in order to get approximately efficient smooth mechanisms one needs
to approximately tie threshold bids to the payments.

\begin{defn} A direct mechanism is $c$-threshold approximate, for some $c\geq 0$, if for any feasible allocation $x\in \X$ and any reported valuation profile $\tilde{v}$:
\begin{equation}
\sum_i \tau_i(x_i,\tilde{v}_{-i})\leq c\sum_i P_i(\tilde{v})
\end{equation}
\end{defn}

The above relation between threshold bids and payments is in the essence of
the analysis of Lucier and Borodin \cite{Lucier2010} as described in the next section.

As an example, consider a first price single item auction setting.  Consider a bid profile $\tilde{v}$ and a feasible allocation where player $i$ wins. The threshold payment for player $i$ to win the auction is exactly equal to the payment that the winner was paying under bid profile $\tilde{v}$. The rest of the players are not allocated hence their threshold payment is $0$. Hence, we observe that a single item first price auction is a $1$-threshold approximate mechanism.

First, we show that if a mechanism is $c$-threshold approximate then this implies a good efficiency guarantee on
the Bayes-Nash Equilibria of the game that it induces.

\begin{theorem}\label{thm:threshold-efficiency}
If a direct mechanism is $c$-threshold approximate and individually rational then it is
$$\left(\beta\left(1-e^{-1/\beta}\right),\beta c\right)\text{-smooth}$$ for any $\beta\geq 0$.
Moreover, it is also conservatively smooth.
\end{theorem}
\begin{proof}
Consider an instance of the valuation profile $v$ and a bid profile $\tilde{v}$. Suppose that bidder $i$ submits a
single-minded value $\theta$ for his optimal allocation $x_i^*(v)$. If
$\tau_i(x_i^*(v),\tilde{v}_{-i})\geq t$ the agent doesn't get allocated $x_i^*(v)$. Otherwise
he is allocated and pays $P_i(t,\tilde{v}_i)$. Since the mechanism satisfies
ex post individual rationality this payment cannot be more than $t$. Otherwise
a player with true value $t$ for $x_i^*(v)$ would be getting negative utility.

Thus the utility of a player from this deviation is at least:
\begin{align*}
u_i(\theta,\tilde{v}_{-i})\geq (v_i(x_i^*(v))-t)\mathds{1}_{t>\tau_i(x_i^*(v),\tilde{v}_{-i})}
\end{align*}
Now a player by using a randomized $\Theta$ that follows a distribution with density
$f(\theta)=\frac{\beta}{v_i(x_i^*(v))-\theta}$ for $\theta \in [0,v_i(x_i^*(v))(1-e^{-1/\beta})]$
he will get expected utility at least:
\begin{align*}
u_i(\Theta,\tilde{v}_{-i})\geq~& \int_{\tau_i(x_i^*(v),\tilde{v}_{-i})}^{v_i(x_i^*(v))(1-e^{-1/\beta})}(v_i(x_i^*(v))-t)f(\theta)d\theta\\
\geq~&\beta\left(1-e^{-1/\beta}\right)v_i(x_i^*(v))-\beta\tau_i(x_i^*(v),\tilde{v}_{-i})
\end{align*}
Adding over all players and using the $c$-threshold payment approximate property we get
the theorem.
\end{proof}

A second price auction on the other hand is not threshold approximate for any $\mu$. The reason is the following: consider a type profile $v$ and a bid profile $\tilde{v}$ where a player with a very small value bids a huge number $H$ and the rest of the players bid truthfully. Then the threshold payment
for the highest value player is $H$, while the payment that the auction receives is of the order of the values
of the rest of the players. Thus payments and threshold payments are unrelated. The latter is the crucial difference between first price and second price payment rules and the reason why we need to employ no-overbidding assumptions
to give efficiency guarantees for second-price payment schemes. Similar to how $c$-threshold approximate
direct mechanisms are connected to smoothness, the following property of weak $c$-threshold approximate mechanisms is
connected to weak smoothness:
\begin{defn} A direct mechanism is weakly $(c_1,c_2)$-threshold approximate, for some $c_1,c_2\geq 0$, if for any feasible allocation $x\in \X$ and any reported valuation profile $\tilde{v}$:
\begin{equation}
\sum_i \tau_i(x_i,\tilde{v}_{-i})\leq c_1\sum_i P_i(\tilde{v}) + c_2\sum_i B_i(\tilde{v}_i,X_i(\tilde{v}))
\end{equation}
\end{defn}

\begin{theorem}\label{thm:weak-threshold-efficiency}
If a direct mechanism is $(c_1,c_2)$-threshold approximate and individually rational then it is  weakly and conservatively
$$\left(\beta\left(1-e^{-1/\beta}\right),\beta c_1,\beta c_2\right)\text{-smooth}$$ for any $\beta\geq 0$.
\end{theorem}
The proof is similar to that of Theorem \ref{thm:threshold-efficiency} and is omitted.

\textbf{Greedy Direct Combinatorial Auctions} A very interesting instance of $c$-threshold approximate mechanisms in the literature is that of Greedy Direct Auctions introduced by Lucier and Borodin \cite{Lucier2010}. In the terminology that we introduced in the previous section, Lucier and Borodin \cite{Lucier2010} proved that in any direct combinatorial auction setting if the allocation is decided by a greedy $c$-approximate mechanism then coupling the mechanism with a first price payment rule we get a $c$-threshold approximate mechanism. These proofs don't assume anything about the valuation of a player and allow for complements.

We note that not all greedy algorithms adhere to the framework defined by Lucier and Borodin \cite{Lucier2010}.
In Section \ref{SEC:GREEDY-MARGINAL} we show how smoothness can capture a version of the $2$-approximation algorithm by Lehmann et al.\cite{Lehmann2001} not captured by \cite{Lucier2010}. The greedy mechanisms studied
in \cite{Lucier2010} are as follows.

\begin{enumerate}
\item Solicit valuation reports $\tilde{v}$.
\item At each iteration pick an (agent, set) pair $(i,S)$ that maximizes a ranking function $r(i,S,\tilde{v}_i(S))$,
and allocate $S$ to $i$.
\item Remove both $i$ and $S$ from consideration and repeat until all items are allocated.
\end{enumerate}
The ranking function is monotone in $S$ (by inclusion) and $v_i(S)$ and could potentially be adaptive
with respect to the existing allocation. For the case of general combinatorial auctions a $\sqrt{k}$-approximate greedy such algorithm exists, where $k$ is the number of items.

Hence, for the setting of greedy first price $c$-approximate combinatorial auctions our framework implies:
\begin{corollary}
Any CE or BNE of a greedy $c$-approximate first price combinatorial auction achieves at least
$\frac{1-e^{-c}}{c}$ of the expected optimal social welfare. If bidders have budgets then it achieves the
same fraction of the optimal effective welfare.
\end{corollary}

Lucier and Borodin \cite{Lucier2010} give a bound of $\frac{1}{c+O(\log(c))}$ for the efficiency of  such a greedy auction. More specifically the bound is $\frac{c-1-\log(c)}{c(c+1+\log(c))}$. Our bound is asymptotically same, but is strictly better. Our bound decreases as $\frac{1}{c+O(c e^{-c})}$ rather than $\frac{1}{c+O(\log(c))}$.
When $c=1$ the bound coincides with the bound of $1-\frac{1}{e}$ for the first price single item auction and our bound is always larger than $\frac{1}{c+\frac{1}{e-1}} \approx \frac{1}{c+0.58}$ and thereby decreases exactly linearly with $c$.

Our composability framework gives new results for the case when several greedy combinatorial auctions are run
simultaneously or sequentially.

\begin{corollary}[Simultaneous with Budgets]
If we run $m$ greedy $c$-approximate first price combinatorial auctions simultaneously and bidders have budgets
and fractionally subadditive valuations across mechanisms, then any CE and any BNE achieves at
least $\frac{1-e^{-c}}{c}$ of the expected optimal effective welfare.
\end{corollary}

\begin{corollary}[Sequential]
If we run $m$ greedy $c$-approximate first price combinatorial auctions sequentially and bidders have
unit-demand valuations across mechanisms, then any CE and any BNE achieves at
least $\frac{1-e^{-c}}{c+1}$ of the expected optimal social welfare.
\end{corollary}

Recall that fractional subadditivity across mechanisms does not impose any assumption on the valuations within each greedy combinatorial auction. Hence, the valuations of the bidders could have complements within the
items sold in each greedy auction, as long as they don't have complements across items sold in different auctions.

Lucier and Borodin \cite{Lucier2010} also examine a second-price type of payment scheme where each agent is charged
his threshold bid for the allocation that he is awarded. In such a mechanism the willingness-to-pay for an allocation
is exactly a player's bid for that allocation. They show that
the mechanism is weakly $(0,c)$-threshold approximate. By applying theorem \ref{thm:weak-threshold-efficiency} for $\beta = 1$ we get an efficiency guarantee $\frac{\beta(1-e^{-1/\beta})}{\beta c+1}=\frac{1-1/e}{c+1}$ for an individual greedy mechanism and for simultaneous and sequential composition. Instead
of using the generic smoothness result of Theorem \ref{thm:weak-threshold-efficiency}, by reinterpreting the techniques
of Lucier and Borodin \cite{Lucier2010} we can easily show that this mechanism is actually weakly $(1,0,c)$-smooth, by considering
the deviation where each player switches to single-mindedly bidding his true valuation on his optimal set \footnote{Since the mechanism charges threshold bids, the utility of a player from this deviation is at least $v_i(x_i^*(v))-\tau_i(x_i^*(v),\tilde{v}_{-i})$.
Summing over all players and using the fact that the greedy mechanism is $(0,c)$-threshold
approximate we get the weak $(1,0,c)$-smoothness result.}.
This gives the slightly better efficiency guarantee of $\frac{1}{c+1}$.

\subsection{Position Auctions}
In this section we consider a generalized version of position auctions
introduced by Abrams et al \cite{Abrams2007}, which allows us to extend analysis of ad auctions to simultaneous and sequential composition as well as to the case where players have budget constraints. It also allows us to capture settings where bidders have
values not only per-click but also per-impression, which is considered an interesting direction from a practical perspective since many companies on the web strive mainly for impressions rather than clicks. We also propose new simple mechanisms that are approximately efficient and robust in terms of simultaneous and sequential composition and in terms of budget constraints.

Consider a setting where the outcome space is an allocation of $n$ positions to $n$ agents. Each agent has
a valuation $v_{ij}$ for being allocated position $j$ and such that the valuations of all the agents are monotone decreasing: if $j\leq j'$ then $v_{ij}\geq v_{ij'}$. The value $v_{ij}$ could be thought of as the value of player
$i$ for appearing at position $j$. This value could consist of a per-click part and a per-impression part.
For instance, if the bidder thinks that his click-through-rate at position $j$ is $a_{ij}$ and knows
that his value per-click is $v_i^{\text{c}}$, while he also has a value $v_{ij}^{\text{im}}$ for appearing at
position $j$, then his valuation for position $j$ is: $v_{ij}=a_{ij} v_i^{\text{c}} + v_{ij}^{\text{im}}$. We just assume that the above total value is monotone in position.

Observe that in our framework notation the allocation space $\X$ consists of vectors $x = (j_1,\ldots,j_n)$ such that $j_i\neq j_{i'}$ for all $i\neq i'$. In addition the allocation space of each player $\X_i=\{1,\ldots,n\}$ is totally ordered from his perspective (in that any outcome where he gets a higher position is greater than any outcome where he gets a lower one) and the value of a player is monotone with respect to his own ordering of allocations.

A position mechanism $M$ defines the action space of the players. We will consider here
mechanisms where players submit only a single bid $b_i$ (interpreted differently by
the different mechanisms that we consider). Given a bid profile, the allocation function of a position mechanism is to assign a position $j_i(b)$ to each player $i$
and the payment of each player is some function of the bid profile $P_i(b)$.

We show that for the class of position-monotone valuations a greedy first price pay-per-impression mechanism (Mechanism \ref{mech:first-pay-per-impression}) is $(\frac{1}{2},1)$-smooth and its second price analog is weakly $(\frac{1}{2},0,1)$-smooth.

\renewcommand{\algorithmcfname}{MECHANISM}
\begin{algorithm}[h]\label{mech:first-pay-per-impression}
\SetKwInOut{Input}{Input}\SetKwInOut{Output}{Output}
\BlankLine
\nl  Solicit a single bid $b_i$ from each player $i$\;
\nl  Order the players according to bids\;
\nl  Allocate positions to players in the order of the bids (i.e. the highest bidder gets the first position, etc.)\;
\nl Charge each player his bid $b_i$
\caption{Greedy first price pay-per-impression mechanism.}
\label{mech:position-per-impression}
\end{algorithm}
\renewcommand{\algorithmcfname}{ALGORITHM}

 \begin{lemma}\label{lem:position-monotone}
 Mechanism \ref{mech:first-pay-per-impression} is conservatively $\left(\frac{1}{2},1\right)$-smooth when valuations are monotone in the position.
\end{lemma}
\begin{proof}
 Consider a valuation profile $v$ and any bid profile $b$ and let $j^*_i$ be the
 optimal position of player $i$ under valuation profile $v$. Suppose that player $i$ deviates to bidding according to the uniform distribution $U[0,v_{ij^*_i}]$. For a given bid profile $b$, let $\pi(j)$ be the player allocated at position $j$. If the
 bid that the player submits is greater than $b_{\pi(j_i^*)}$ then he is allocated position at least as high as $j_i^*$.
 By monotonicity of valuations with respect to position we get that his utility from the deviation is at least:
\begin{align*}
 u_i(b_i',b_{-i})\geq~& \int_{0}^{v_{ij^*_i}} \left(v_{ij^*_i}\cdot \mathbf{1}_{\{b_{\pi(j^*_i)}<t\}} -t\right)f(t)dt\\
\geq~& \int_{b_{\pi(j^*_i)}}^{v_{ij^*_i}} v_{ij^*_i}f(t)dt-\int_0^{v_{ij^*_i}}t f(t)dt\\
=~& v_{ij^*_i}-b_{\pi(j^*_i)}-\frac{v_{ij^*_i}}{2}=\frac{v_{ij^*_i}}{2}-b_{\pi(j^*_i)}
\end{align*}
where $f(x)=1/v_{ij^*_i}$ is the density function. Summing over all players we get the theorem.
\end{proof}

The fact that the above mechanism is smooth for any monotone valuation allows us to invoke
Theorem \ref{thm:monotone-xos} and get composability results. In addition the fact that the
class of monotone valuations is closed under cappings allows us to invoke our budget constraint
results.

\begin{corollary}[Simultaneous with Budgets]
If we run $m$ greedy first price pay-per-impression mechanisms simultaneously and bidders have monotone fractionally subadditive valuations and budget constraints then any CE and BNE achieves at least $1/2$ of the expected optimal effective welfare.
\end{corollary}

\begin{corollary}[Sequential]
If we run $m$ greedy first price pay-per-impression mechanisms sequentially and bidders have unit-demand
valuations then every CE and BNE achieves at least $1/4$ of the expected optimal welfare.
\end{corollary}

Note that in the last theorem, unit-demand valuations in our terminology, means
that a players value for getting several impressions at different position mechanisms is of the form: $$v_i(j_i^1,j_i^2,\ldots,j_i^m) = \max_{k\in [m]} v_i^k(j_i^k),$$ where the induced valuations $v_i^k(j_i^k)$ are monotone in the position $j_i^k$ allocated at position mechanism $\M_k$.

\textbf{Threshold-Price Mechanism \ref{mech:first-pay-per-impression}.} We also consider the variation of Mechanism \ref{mech:first-pay-per-impression} studied in Abrams et al. \cite{Abrams2007}, where each player is charged the bid of the player in the position beneath him.
We show that such a mechanism is conservatively and weakly $\left(\frac{1}{2},0,1\right)$-smooth, implying a bound of $1/4$ in isolation, when
composed simultaneously under budget constraints and when composed sequentially.

In \cite{Abrams2007} it was shown that in the full information setting there will always exist a Pure Nash Equilibrium of this mechanism
that achieves optimal social welfare, thereby generalizing the result of Edelman et al \cite{Edelman2007} where only valuations per-click
where considered. However, no price of anarchy analysis exists for this mechanism and the Bayesian setting or solution concepts
that use randomization have not been studied.

First, we clarify our no-overbidding assumption for the mechanism of \cite{Abrams2007}. In this mechanism
when a player is allocated position $j$ with a bid $b_i$ then his maximum willingness-to-pay is $b_i$, since
in the strategy profile where the player in position $j+1$ bids $b_i$ too, he is charged $b_i$. Thus under
Definition \ref{def:willingness-to-pay} of willingness-to-pay we have:
$$B_i(b_i,j) = b_i$$
Using a proof identical to that of Lemma \ref{lem:position-monotone} we can show the weak and conservative smoothness of this mechanism.

\begin{lemma} The threshold price version of Mechanism \ref{mech:first-pay-per-impression} where each player is charged the
bid in the position beneath him is weakly and conservatively $\left(\frac{1}{2},0,1\right)$-smooth.
\end{lemma}

Our no-overbidding assumption states that in expectation no player is bidding more than his value for the expected
position he gets. A randomized bid profile ${\bf b}$ satisfies the no-overbidding assumption if:
$$\E_{{\bf b}_i}[{\bf b}_i] \leq \E_{{\bf b}}[v_{ij_i({\bf b})}]$$
If a player participates in many position auctions his strategy is to submit a bid $b_i^k$ at each position auction $\M_k$.
Let $b_i=(b_i^k)_{k\in [m]}$ and $b^k=(b_i^k)_{i\in [n]}$. Then the no-overbidding assumption generalizes to:
$$\E_{{\bf b}_i}\left[\sum_{k\in [m]}{\bf b}_i^k\right] \leq \E_{{\bf b}}\left[v_i(j_i^1({\bf b}^1),\ldots,j_i^m({\bf b}^m))\right]$$
Under this no-overbidding assumption, our framework gives the following results.

\begin{corollary}[Simultaneous with Budgets]
If we run $m$ greedy threshold price pay-per-impression mechanisms simultaneously and bidders have monotone fractionally subadditive valuations and budget constraints then any CE and BNE that satisfies the no-overbidding assumption, achieves at least $1/4$ of the expected optimal effective welfare.
\end{corollary}

\begin{corollary}[Sequential]
If we run $m$ greedy threshold price pay-per-impression mechanisms sequentially and bidders have unit-demand
valuations then every CE and BNE that satisfies the no-overbidding assumption, achieves at least $1/4$ of the expected optimal welfare.
\end{corollary}

\subsubsection{Per-Click Valuations}
To draw a stronger connection with existing position auction literature we now examine the case when bidders have only
valuations per-click and not per impression.  We will consider two special cases of bidder valuations:
\begin{enumerate}
\item $v_{ij}=a_{ij}\tilde{v}_{ij}$: click-through-rates of player $i$ at position $j$ depend on both $i$ and $j$ in a non-separable way and players have position specific per-click valuations.
\item $v_{ij}=a_{ij} \tilde{v}_i$:  per-click valuations are position independent
\end{enumerate}


While this class of valuations neglects effects captured by the more general
valuation models, special cases of this model are widely used in the literature. The latter case contains the separable model that has been long studied in the algorithmic game theory literature and has become the standard \cite{Edelman2007,Caragiannis2012,Lucier2011}.

\begin{defn} We say that the click through rates are separable, when
 $a_{ij}=\alpha_j \gamma_i$ for all $i$ and $j$, that is, the click through rate is the product of a factor depending on the slot and a factor depending on the advertiser.
\end{defn}

We use our smoothness framework to strengthen results in the literature.
We give a simple smooth mechanism for the first case, which is equivalent to 
the standard form of the Generalized First Price (GFP) auction when specialized to the case of separable click-through rates $a_{ij}=\alpha_j\gamma_i$, showing an $1/2$ efficiency bound on GFP and its generalization to the first case above, and an $1/4$ efficiency bound for the corresponding second price analog.

For the second case, we show that the above auction is $(1-1/e,1)$ smooth when using first price and weakly $(1-1/e,0,1)$ smooth when using second price. This result generalizes the efficiency bound of $\frac{1}{2}(1-\frac{1}{e})$ of Caragiannis
et al \cite{Caragiannis2012} that considered only the separable case when $a_{ij}=\alpha_j\gamma_i$.

Note, however, that this class of valuations is not closed under capping, so our results do not extend to the case with budgets. The smoothness results
we provide in the remainder of the section do imply efficiency guarantees in isolation and for special cases of complement-free valuations (e.g. bidders have value only for the $k$ highest impressions they got and their value per impression is of the form for which smoothness is proved).


\textbf{Variable Click Value.} First we consider the case when each bidder $i$ has a click-through-rate $a_{ij}\in [0,1]$ when he occupies position $j$ and a value $\tilde{v}_{ij}$ when he receives a click at position $j$. We assume that $a_{ij}$ and $\tilde{v}_{ij}$ are both decreasing in $j$. In addition, players submit per click bids and thereby are charged $a_{ij} b_i$.
Hence, a player's utility at bid profile $b$ is:
\begin{equation}
 u_i(b)=a_{ij_i(b)}(\tilde{v}_{ij_i(b)}-b_i)
\end{equation}
The utility of a player is quasi-linear with value $v_{ij}=a_{ij}\tilde{v}_{ij}$ and payment scheme $P_i(b)=a_{ij_i(b)}b_i$.

\renewcommand{\algorithmcfname}{MECHANISM}
\begin{algorithm}[h]
\SetKwInOut{Input}{Input}\SetKwInOut{Output}{Output}
\BlankLine
\nl  Solicit a single bid $b_i$ from each player $i$\;
\nl  Allocate position $1$ to the bidder $i_1=\arg\max_i a_{i1}b_i$. Allocate position $2$ to bidder $i_2=\arg\max_{i\neq i_1} a_{i2}b_i$, etc.\;
\nl  If player $i$ is allocated to position $j$ charge him $a_{ij} b_i$.
\caption{Greedy first price pay-per-click position mechanism for non-separable click-through-rates.}
\label{mech:position-per-click-non-sep}
\end{algorithm}
\renewcommand{\algorithmcfname}{ALGORITHM}

\begin{lemma}\label{lem:non-separable}
 Mechanism \ref{mech:position-per-click-non-sep} is $\left(\frac{1}{2},1\right)$-smooth when click-through-rates and valuations per click are monotone in the position.
\end{lemma}
\begin{proof}
 Consider a valuation profile $v$ and a bid profile $b$ and let $j^*_i$ be the optimal position of player $i$ under valuation profile $v$. Suppose that player $i$ deviates to bidding according to the uniform distribution $U[0,\tilde{v}_{ij^*(i)}]$. Then his utility from the deviation is:
\begin{align*}
 u_i(b_i',b_{-i})=~& \int_{0}^{\tilde{v}_{ij^*_i}} a_{ij_i(t,b_{-i})}(\tilde{v}_{ij_i(t,b_{-i})}-t)f(t)dt
\end{align*}
where $f(x)=1/\tilde{v}_{ij^*(i)}$ is the density function.
Let $\pi(j^*_i)$ be the player allocated at position $j^*_i$ under bid profile $b$. Observe that if $a_{ij^*_i}t\geq a_{\pi(j^*_i)j^*_i}b_{\pi(j^*_i)}$ then if player $i$ hasn't already been allocated
a higher position he will definitely win position $j^*_i$. Thereby, for $t$ in the above range we know that player $i$ will get a position higher than or equal to $j^*_i$.

Since $\tilde{v}_{ij}$ is decreasing in $j$ and $t\in[0,\tilde{v}_{ij^*_i}]$ for
$$t\geq \frac{a_{\pi(j^*_i)j^*_i}}{a_{ij^*_i}}b_{\pi(j^*_i)}=\tau_i$$
we have that $\tilde{v}_{ij_i(t,b_{-i})}-t$ is positive. By the monotonicity of $a_{ij}$ with respect to $j$ we have:
\begin{align*}
 u_i(b_i',b_{-i})\geq& \int_{\tau_i}^{\tilde{v}_{ij^*_i}} a_{ij^*_i}(\tilde{v}_{ij^*_i}-t)f(t)dt-\int_{0}^{\tau_i}a_{ij_i(t,b_{-i})} t f(t)dt\\
\geq~& \int_{\tau_i}^{\tilde{v}_{ij^*_i}} a_{ij^*_i}(\tilde{v}_{ij^*_i}-t)f(t)dt
 -\int_{0}^{\tau_i}a_{ij^*_i} t f(t)dt\\
=~& \int_{\tau_i}^{\tilde{v}_{ij^*_i}} a_{ij^*_i}\tilde{v}_{ij^*_i}f(t)dt
-\int_{0}^{v_{ij^*_i}}a_{ij^*_i} t f(t)dt\\
=~& a_{ij^*_i}\tilde{v}_{ij^*_i}-a_{ij^*_i}\tau_i - \int_{0}^{\tilde{v}_{ij^*_i}}a_{ij^*_i}tf(t)dt\\
=~& a_{ij^*_i}\tilde{v}_{ij^*_i}-a_{\pi(j^*_i)j^*_i}b_{\pi(j^*_i)}-a_{ij^*_i}\frac{\tilde{v}_{ij^*_i}}{2}\\
=~& \frac{a_{ij^*_i}\tilde{v}_{ij^*_i}}{2}-a_{\pi(j^*_i)j^*_i}b_{\pi(j^*_i)}
\end{align*}
By summing over all players we get the theorem.
\end{proof}

\textbf{Separable CTRs}. Observe that when the click-through-rates are separable, then Mechanism \ref{mech:position-per-click-non-sep}
takes the standard form of the Generalized First Price (GFP) auction that has been studied in the literature. Specifically,
the allocation function of Mechanism \ref{mech:position-per-click-non-sep} can be concisely described as: weight each players bid by his quality factor and allocate positions in order of the weighted bid. Each player is then charged his bid, per-click: $a_{j_i(b)} \gamma_i b_i$.
Hence, the utility of a player at some bid profile is:
\begin{equation}
 u_i(b)=a_{j_i(b)}(\tilde{v}_{ij_i(b)}-b_i)
\end{equation}

The specialization of Lemma \ref{lem:non-separable} for separable click-through-rates gives a $\frac{1}{2}$ efficiency bound
for the Generalized First Price auction even when the per-click valuations of the players.
\begin{corollary}
The Generalized First Price Auction is $\left(\frac{1}{2},1\right)$-smooth when click-through-rates are separable
and valuations per-click are dependent on the position.
\end{corollary}

\textbf{Position independent value-per-click.} Better smoothness properties can be derived if the value per-click of a player is the same for all positions (denoted by $\tilde{v}_i$), even when click-through-rates are not separable.

%


\begin{lemma}\label{lem:position-indep}
 Mechanism \ref{mech:position-per-click-non-sep} is $\left(1-\frac{1}{e},1\right)$-smooth when per-click valuations are position independent  even if click-through-rates are not separable.
\end{lemma}
\begin{proof}
 Suppose that player $i$ deviates to bidding according to distribution with density
 function $f(t)=\frac{1}{\tilde{v}_i-t}$ and support $[0,(1-\frac{1}{e})\tilde{v}_i]$. As  stated in the proof of Lemma \ref{lem:non-separable} if
 $$t\geq \frac{a_{\pi(j^*_i)j^*_i}}{a_{ij^*_i}}b_{\pi(j^*_i)}=\tau_i$$
 (where $\pi(j)$ is the player at position $j$ in the current bid profile) then player $i$ is assigned a position at least as high as $j^*_i$.
 Thus his utility from the deviation is:
\begin{align*}
 u_i(b_i',b_{-i})=~&\int_0^{(1-\frac{1}{e})\tilde{v}_i} a_{ij_i(t,b_{-i})}(\tilde{v}_i-t)f(t)dt\\
\geq~& \int_{\tau_i}^{(1-\frac{1}{e})\tilde{v}_i} a_{ij^*_i}(\tilde{v}_i-t)f(t)dt\\
=~& \left(1-\frac{1}{e}\right)a_{j^*_i}\tilde{v}_i-a_{\pi(j^*_i)}b_{\pi(j^*_i)}
\end{align*}
By summing over all players we get the theorem.
\end{proof}

If the click-through-rates are separable, then this brings us to the standard model studied in the literature, where the valuation
of a player $i$ from being assigned at position $j$ is: $a_j \gamma_i v_i$ and thereby the utility of a player at some
bid profile is:
\begin{equation}
u_i(b) = a_{j_i(b)}\gamma_i (\tilde{v}_i-b_i)
\end{equation}
The specialization of Lemma \ref{lem:position-indep}, gives a better smoothness property for the Generalized First Price Auction.
\begin{corollary}
 The GFP auction is $\left(1-\frac{1}{e},1\right)$-smooth when per-click valuations are position independent and click-through-rates are separable.
\end{corollary}

\textbf{Threshold Price Mechanism \ref{mech:position-per-click-non-sep}.} We consider a threshold-price version of Mechanism \ref{mech:position-per-click-non-sep}
where a player is charged, per-click, the minimum bid he had to make to get his position. To draw a strong connection with
existing position auction literature we will analyze this mechanism only on the special case where click-through-rates
are separable and valuations are position independent: $v_{ij}=a_j \gamma_i \tilde{v}_i$.

Under this valuation model
the threshold-price version of Mechanism \ref{mech:position-per-click-non-sep} becomes the standard Generalized Second Price
(GSP) auction introduced by Edelman et al \cite{Edelman2007} and extensively studied from the price of anarchy perspective
\cite{Lucier2011,Caragiannis2012}. In this mechanism, under strategy profile $b$, each player $i$ is charged
$\frac{\gamma_{\pi(j_i(b)+1)} b_{j_i(b)+1}}{\gamma_i}$ per-click, where $\pi(j)$ is the player that got position $j$ under bid profile
$b$, and thus his utility at some bid profile is:
\begin{equation}
u_i(b) = a_j \gamma_i\left(\tilde{v}_i - \frac{\gamma_{\pi(j_i(b)+1)} b_{j_i(b)+1}}{\gamma_i}\right)
\end{equation}

In this mechanism a player's willingness-to-pay for a position $j$ is simply $a_j \gamma_i b_i$ since in the special
case where the player beneath him was bidding $\frac{\gamma_i}{\gamma_{\pi(j_i(b)+1)}}b_i$, player $i$ is charged
an expected total payment of $a_j\gamma_i b_i$. Thus from Definition \ref{def:willingness-to-pay} of willingness-to-pay we have that:
$$B_i(b_i,j) = a_j \gamma_i b_i$$
Our no-overbidding assumption will then become:
\begin{align*}
\E_{{\bf b}}[a_{j_i(b)}{\bf b}_i] \leq \E_{{\bf b}}[a_{j_i(b)}\tilde{v}_i]
\end{align*}
Caragiannis et al \cite{Caragiannis2012} use a point-wise no-overbidding assumption that for any
bid in the support of a player's strategy $b_i \leq \tilde{v}_i$ and prove that such an overbidding is weakly dominated.
That assumption would imply our weaker in expectation assumption.

Under the no-overbidding assumption and using the same proof as in Lemma \ref{lem:position-indep} specialized for
separable click-through-rates would give that the Generalized Second Price auction is
weakly $\left(1-\frac{1}{e},0,1\right)$-smooth. This implies the efficiency result of $\frac{1}{2}\left(1-\frac{1}{e}\right)$ that was given in Caragiannis et al \cite{Caragiannis2012} and the proof of Lemma \ref{lem:position-indep} is a generalization of their analysis.
\begin{corollary}
The Generalized Second Price auction is weakly $\left(1-\frac{1}{e},0,1\right)$-smooth when per-click valuations are
position independent and click-through-rates are separable.
\end{corollary}
In fact applying the same proof of Lemma \ref{lem:position-indep} we get a generalization of this result for non-separable click-through-rates.
\begin{corollary}
The threshold-price version of Mechanism \ref{mech:position-per-click-non-sep} is weakly $\left(1-\frac{1}{e},0,1\right)$-smooth
when per-click valuations are position independent, even if the click-through-rates are not separable.
\end{corollary}
This result implies an efficiency bound of $\frac{1}{2}\left(1-\frac{1}{e}\right)$ under the same no-overbidding
assumption that was used by Caragiannis et al \cite{Caragiannis2012}.

\subsection{Public Goods Auctions}
In this section we consider a first price auction for choosing a public good
and show that it is $(\frac{1-e^{-n}}{n},1)$-smooth where $n$ is the number of participants in the mechanism. This bound is proportional to the number of participants
in the auction. We then give an application of a simultaneous public good auction where the number of participants at each one is small. We leave as a very important open question whether there exist smooth mechanisms for the combinatorial public project setting that imply efficiency guarantees that are independent of the number of participants.

We consider the following formal setting: there are $n$ bidders and $m$ public projects. The mechanism wants to choose a single public project to implement and each player $i$ has a value $v_{ij}$ if project $j$ is implemented.

\renewcommand{\algorithmcfname}{MECHANISM}
\begin{algorithm}[h]\label{mech:single-public-good}
\SetKwInOut{Input}{Input}\SetKwInOut{Output}{Output}
\BlankLine
\nl  Solicit bids $b_{ij}$ from each player $i$ for each project $j$\;
\nl  For a project $j\in[m]$, let $B_j=\sum_{i\in [n]} b_{ij}$\;
\nl  Pick project $j(b)=\arg\max_{j\in [m]} B_j$\;
\nl  Charge each player his bid for the chosen project $b_{ij(b)}$
\caption{First price public project auction.}
\end{algorithm}
\renewcommand{\algorithmcfname}{ALGORITHM}

A generalization of Mechanism \ref{mech:single-public-good} where multiple projects are to be chosen and
bidders have combinatorial valuations on the projects, is considered by Singer et al \cite{Singer2012},
who give an efficiency analysis for pure, correlated and Bayes-Nash equilibria. For the special case that
we describe here, their analysis implies an efficiency bound of $\frac{1}{n+1}$. The following lemma
gives a slightly better result.

\begin{theorem}Mechanism \ref{mech:single-public-good} is $(\frac{1-e^{-n}}{n},1)$-smooth.\end{theorem}
\begin{proof}
Consider a valuation profile $v$ and a bid profile $b$. Let $j^*$ be the optimal project for this valuation profile and let $j(b)$ be the project chosen under bid profile $b$. Suppose that each player $i$ deviates to bidding $f(t)=\frac{1/n}{v_{ij^*}-b}$ with support $[0,(1-e^{-n})v_{ij^*}]$ on project $j^*$ only. Let $B_j$ be the total bid of project $j$ under bid profile $b$. Then the utility of player $i$ from the deviation, even if he is the only one bidding on project $j^*$, is at least:
\begin{align*}
 u_i(b_i',b_{-i})\geq~&\int_{B_{j(b)}}^{(1-e^{-n})v_{ij^*}}(v_{ij^*}-t)f(t)dt\\
\geq~& \frac{1}{n} (1-e^{-n})v_{ij^*}-\frac{1}{n} B_{j(b)}
\end{align*}
By summing over all players we get:
\begin{align*}
 \sum_i u_i(b_i',b_{-i})\geq \frac{1}{n}(1-e^{-n})V_j^* - B_{j(b)}
\end{align*}
Which gives the theorem.
\end{proof}

\textbf{Simultaneous Local Public Good Auctions.} Consider a social network setting where players bid for facilities to be placed on nodes in a social network. Each node is an agent and when a facility is placed on a node then all of the neighbors of the node can use it. There exists a set of facilities $F_u$ that can be placed on each node $u$ (let $F_u$ contain also the empty facility for the case where no facility is built). Now we assume that auctioneers run a public good auction on each node to decide which facility they are going to place. Specifically, he asks from the node and its neighbors to submit a bid for each possible facility. Then he is going to choose the facility that received the highest sum of bids and charge each player his bid for the chosen facility.

Each mechanism is $(\frac{1-e^{-d_i}}{d_i},1)$-smooth where $d_i$ is the degree of the node that is auctioned. Now our framework shows that if we run simultaneous such auctions and the valuation of a player is a fractionally subadditve valuation over the facilities placed on his neighboring nodes, then the overall social welfare of this game will be at most $\frac{D}{1-e^{-D}}$ where $D=\max_i d_i$. Similarly, one could imagine of a setting where facilities are not placed on nodes of the graph but rather on edges of it or on hyper-edges in a hyper-graph that tries to model groups of interested agents. In such settings our framework implies that the above simultaneous local public good mechanism has price of anarchy at most $\frac{k}{1-e^{-k}}$, where $k$ is the size of the hyper-edge.

\subsection{Proportional Bandwidth Allocation\\ Mechanism}
In this section we consider the bandwidth allocation setting of Johari and Tsitsiklis \cite{Johari2004}. 
In this setting a bandwidth of $C$ is to be split among $n$ bidders. The bidders submit a bid $b_i$
which they have to pay no matter how much bandwidth they receive. Given the bid profile each player is allocated a bandwidth proportional to his bid:
\renewcommand{\algorithmcfname}{MECHANISM}
\begin{algorithm}[h]\label{mech:bandwidth-allocation}
\SetKwInOut{Input}{Input}\SetKwInOut{Output}{Output}
\BlankLine
\nl  Solicit a single bid $b_i$ from each player $i$\;
\nl  Allocate to player $i$ bandwidth $x_i(b)=\frac{b_i C}{\sum_{j\in N}b_j}$\;
\nl  Charge each player his bid $b_i$
\caption{Proportional bandwidth allocation mechanism.}
\end{algorithm}
\renewcommand{\algorithmcfname}{ALGORITHM}

Each player has a concave value function $v_i(x_i)$ for getting a share of bandwidth $x_i$, with $v_i(0)=0$,
and his utility is quasi-linear with respect to payments:
\begin{equation}
u_i(b)=v_i(x_i(b))-b_i
\end{equation} 
As one can easily observe the latter mechanism falls into our general definition of a mechanism with quasi-linear preferences. We will show that such a mechanism is $(2-\sqrt{3},1)$-smooth. This will imply efficiency
guarantees of approximately $1/4$ for any CE and BNE as well as for simultaneous compositions and sequential
composition of bandwidth allocation mechanisms. For the simultaneous setting it also implies such a bound even
when players have budget constraints. Johari and Tsitsiklis \cite{Johari2004} give an efficiency bound of $3/4$
but their efficiency guarantee is proved only for the case of pure nash equilibria and only in the complete information setting. Hence, though our bound is slightly worse, it is a bound that extends to a plethora of relaxations and extensions.

\begin{lemma}
The proportional bandwidth allocation mechanism is conservatively $(2-\sqrt{3},1)$-smooth when
value functions $v_i:[0,C]\rightarrow \R^+$ are concave and $v(0)=0$.
\end{lemma}
\begin{proof}
Given a valuation profile $v$ for each player, let $x_i^*(v)$ be the bandwidth allocated 
to player $i$ in the optimal allocation. For simplicity we will denote it with $x_i^*$ for the 
remainder of the proof since we focus on a specific valuation profile. 

Consider the deviation where player $i$ deviates to bidding uniformly at random $B_i\sim U[0,\lambda v_i(x_i^*)]$,
for some constant $\lambda$ that will be determined later on. Then his expected utility for any 
bid profile $b_{-i}$ is as follows:
\begin{align*}
u_i(B_i,b_{-i})=\int_{0}^{\lambda v_i(x_i^*)} \frac{v_i(x_i(t,b_{-i}))}{\lambda v_i(x_i^*)}dt - \frac{1}{2}\lambda v_i(x_i^*)
\end{align*}
Given the bids of the rest of the players $b_{-i}$, if player $i$ bids above $\frac{x \sum_{j\neq i}b_j}{C-x}$
then he is given a bandwidth share of at least $x$ for any $x$. Thus for all the $t\geq \frac{x_i^*\sum_{j\neq i}b_j}{\mu C-x_i^*}$ player $i$ is allocated a bandwidth of at least $x_i^*/\mu$.

Thus by monotonicity of $v_i$ his utility from the deviation is at least:
\begin{align*}
u_i(B_i,b_{-i})\geq \int_{\frac{x_i^*\sum_{j\neq i}b_j}{\mu C-x_i^*}}^{\lambda v_i(x_i^*)}\frac{v_i(x_i^*/\mu)}{\lambda v_i(x_i^*)}dt - \frac{1}{2}\lambda v_i(x_i^*)
\end{align*}
By concavity and the fact that $v_i(0)=0$ we know that $v_i\left(\frac{x_i^*}{\mu}\right)\geq \frac{v_i(x_i^*)}{\mu}$ for any $\mu\geq 1$. Thus:
\begin{align*}
u_i(B_i,b_{-i})\geq~& \int_{\frac{x_i^*\sum_{j\neq i}b_j}{\mu C-x_i^*}}^{\lambda v_i(x_i^*)}\frac{v_i(x_i^*)}{\mu\lambda v_i(x_i^*)}dt - \frac{1}{2}\lambda v_i(x_i^*)\\
=~&\int_{\frac{x_i^*\sum_{j\neq i}b_j}{\mu C-x_i^*}}^{\lambda v_i(x_i^*)}\frac{1}{\mu\lambda}dt - \frac{1}{2}\lambda v_i(x_i^*)\\
=~&\frac{1}{\mu}v_i(x_i^*)-\frac{1}{\lambda\mu}\frac{x_i^*\sum_{j\neq i}b_j}{\mu C-x_i^*} - \frac{1}{2}\lambda v_i(x_i^*)
\end{align*}
Since $x_i^*\leq C$ and $\sum_{j\neq i}b_j \leq \sum_j b_j$ we get:
\begin{align*}
u_i(B_i,b_{-i})\geq \frac{1}{\mu}v_i(x_i^*)-\frac{1}{\lambda\mu}\frac{x_i^*\sum_j b_j}{(\mu-1)C} - \frac{1}{2}\lambda v_i(x_i^*)
\end{align*}

Summing over all players we get:
\begin{align*}
\sum_i u_i(B_i,b_{-i})\geq~& \left(\frac{1}{\mu}-\frac{\lambda}{2}\right)\sum_i v_i(x_i^*) - \sum_i \frac{1}{\lambda\mu}\frac{x_i^*\sum_j b_j}{(\mu-1)C}\\
=~&\left(\frac{1}{\mu}-\frac{\lambda}{2}\right)\sum_i v_i(x_i^*) -  \frac{1}{\lambda\mu(\mu-1)}\sum_j b_j
\end{align*}
By setting $\lambda=\frac{1}{\mu(\mu-1)}$ we get that the mechanism is $(\frac{1}{\mu}-\frac{\lambda}{2},1)$-smooth.
By optimizing over $\mu$ we get that the best bound is implied by $\mu=\frac{1}{2}(3+\sqrt{3})$ for which we get that
the mechanism is $(2-\sqrt{3},1)$-smooth.
\end{proof}

Now observe that the valuation space for which smoothness is proved is the set of all concave valuations
on $\R^+$, with $v(0)=0$. Clearly, $\R^+$ is a distributive lattice. Submodularity with respect to $\R^+$
simply means concavity. By theorem \ref{thm:lattice-xos} we know that if several such bandwidth allocation mechanisms happen simultaneously and the valuation of the player is submodular on the product lattice, then we can express such a valuation with induced valuations that are capped marginals:
\begin{equation*}
v_j(\xsij) = v(\xsij\wedge \tilde{x}_i^j,\tilde{x}_i^{-j})-v(0,\tilde{x}_i^{-j})
\end{equation*}
Those function are concave with respect to $\R^+$ and have $v(0)=0$. Therefore we can apply our simultaneous
composition theorem for any submodular valuation with respect to the lattice $\R^m_+$. If the valuations
are continuously differentiable then submodularity on $\R^m_+$ simply means that:
$\frac{\partial^2 v(x_i)}{(\partial x_i^j)^2}\leq 0$ and $\frac{\partial^2 v(x_i)}{\partial x_i^j\partial x_i^{j'}}\leq 0$. Thus the function is concave coordinate-wise and has decreasing differences.

In addition observe that even if we cap a submodular valuation on $\R^+$ then it remains submodular. Thus
we can also invoke our budget constraint theorems. 

\begin{corollary}[Simultaneous with Budgets] If we run $m$ simultaneous bandwidth allocation mechanisms and the valuations are submodular on $\R^m_+$ and bidders have budgets then any CE and any BNE achieve at least $\frac{1}{2-\sqrt{3}}\approx \frac{1}{3.73}$ of the expected optimal effective welfare.
\end{corollary}

For sequential composition we get our theorem for the case where the valuation of the bidder is the maximum
among his valuations on different links: $v_i(x_i)=\max_{j} v_{ij}(x_{ij})$.
\begin{corollary}[Sequential]
If we rum $m$ sequential bandwidth allocation mechanisms and the valuations are unit-demand then
any CE and any BNE achieves at least $\frac{1}{2(2-\sqrt{3})}$ of the expected optimal social welfare.
\end{corollary}

\subsection{Multi-Unit Auctions with Concave Values}
Consider the following setting: An auctioneer wants to sell $k$ units of a good. A bidder's valuation
is an increasing concave function $v_i(j)$ of the amount $j$ of goods he gets. We consider the following auction:

\renewcommand{\algorithmcfname}{MECHANISM}
\begin{algorithm}[h]\label{mech:first-price-multi-unit}
\SetKwInOut{Input}{Input}\SetKwInOut{Output}{Output}
\BlankLine
\nl  Solicit bids $b_{i1},\ldots,b_{ik}$ for marginal values from each player $i$ which are restricted to be decreasing\;
\nl  At each iteration allocate the extra unit to the bidder that has the maximum marginal
bid for getting it, conditional on the items he has already been allocated\;
\nl  Repeat until all units are allocated or until no player has value for an extra unit\;
\nl  If player $i$ is allocated $k_i$ units then charge him  $\sum_{j=1}^{k_i}b_{ij}$.
\caption{Greedy First-Price Multi-Unit Auction with concave values.}
\end{algorithm}
\renewcommand{\algorithmcfname}{ALGORITHM}

We will denote with $k_i(b)$ the units allocated to bidder $i$ under bid profile $b$. We will also denote with
$p_j(b)$ to be the $j-th$ lowest price for which a unit was sold by the algorithm, i.e. the bid of the $j$-th from the
end unit that was sold. The utility of a bidder is still quasi-linear with money:
\begin{equation}
u_i(b) = v_i(k_i(b)) - \sum_{j=1}^{k_i(b)}b_{ij}
\end{equation}

We show that the greedy multi-unit auction is $(\frac{1}{2}\left(1-\frac{1}{e}\right),1)$-smooth thereby implying an efficiency guarantee of $\frac{e-1}{2e}\approx 1/3.16$ when studied in isolation.

\begin{lemma}\label{lem:greedy-multi-unit}
Mechanism \ref{mech:first-price-multi-unit} is conservatively
$$\left(\frac{1}{2}\left(1-\frac{1}{e}\right),1\right)$$ smooth
when bidders' valuations $v_i:\N\rightarrow \R^+$ are concave with $v_i(0)=0$.
\end{lemma}
\begin{proof}
Suppose that bidder $i$ deviates to stating that his $k_i^*$ highest marginal valuations
are all $t$ for some randomly drawn $t$ according to the
distribution with probability density function $f(t)=\frac{1}{\frac{v_i(k_i^*)}{k_i^*}-t}$ and support $[0,\frac{v_i(k_i^*)}{k_i^*}\left(1-\frac{1}{e}\right)]$.
For his remaining marginal valuations he bids $0$. Then the utility of player $i$ from this deviation is:
\begin{align*}
u_i(b_i',b_{-i})=~&\int_0^{\frac{v_i(k_i^*)}{k_i^*}\left(1-\frac{1}{e}\right)} (v_i(k_i(t,b_{-i}))-k_i(t,b_{-i})t)f(t)dt\\
=~& \int_0^{\frac{v_i(k_i^*)}{k_i^*}\left(1-\frac{1}{e}\right)} k_i(t,b_{-i})\left(\frac{v_i(k_i(t,b_{-i}))}{k_i(t,b_{-i})}-t\right)f(t)dt
\end{align*}

Since bidder $i$ bids positive only for his $k_i^*$ highest marginals we know that he is allocated
at most $k_i^*$ units. Hence, $k_i(t,b_{-i})\leq k_i^*$ for all $t$. In addition by concavity
we know that for any $k\leq k_i^*: \frac{v(k)}{k}\geq \frac{v(k_i^*)}{k_i^*}$. Hence:
\begin{align*}
u_i(b_i',b_{-i})\geq~& \int_0^{\frac{v_i(k_i^*)}{k_i^*}\left(1-\frac{1}{e}\right)} k_i(t,b_{-i})\left(\frac{v_i(k_i^*)}{k_i^*}-t\right)f(t)dt\\
=~& \int_0^{\frac{v_i(k_i^*)}{k_i^*}\left(1-\frac{1}{e}\right)} k_i(t,b_{-i})dt
\end{align*}
For any $j\in [1,k_i^*]$, if $t>p_j(b)$ then $k_i(t,b_{-i})\geq j$. Hence, we have:
\begin{align}
u_i(b_i',b_{-i})\geq~&\int_{p_j(b)}^{\frac{v_i(k_i^*)}{k_i^*}\left(1-\frac{1}{e}\right)} jdt\nonumber\\
=~&\frac{j}{k_i^*}\left(1-\frac{1}{e}\right)v_i(k_i^*)-jp_j(b)\label{eqn:lower_bound_mpa}
\end{align}

Now we need to find the right pick of $j$ in our analysis, such that when adding the above inequality for all players
then the negative part on the right hand side will be the total revenue of the auction at bid profile $b$.

Since prices  $p_j(b)$ are increasing in $j$ we know that:
\begin{align*}
j p_j(b) \leq \sum_{t=0}^{j-1} p_{j+t}(b)
\end{align*}
If we choose a $j$ such that $2j-1\leq k_i^*$ then:
\begin{align*}
j p_j(b) \leq \sum_{t=0}^{j-1} p_{j+t}(b)\leq \sum_{t=j}^{k_i^*} p_{t}(b)\leq \sum_{t=1}^{k_i^*}p_t(b)
\end{align*}
Observe that since $k_i^*$ are integers, if we choose $j=\lceil \frac{k_i^*}{2} \rceil$ then
we know that $j\leq \frac{k_i^*+1}{2}$ and therefore $2j-1\leq k_i^*$. Thus if we apply
Inequality \eqref{eqn:lower_bound_mpa} for $j=\lceil \frac{k_i^*}{2} \rceil$ we get:
\begin{align*}
u_i(b_i',b_{-i})\geq \frac{1}{2}\left(1-\frac{1}{e}\right)v_i(k_i^*)-\sum_{t=1}^{k_i^*}p_t(b)
\end{align*}
Last observe that since $\sum_i k_i^*=k$ and prices  $p_t(b)$ are increasing in $t$: $\sum_i \sum_{t=1}^{k_i^*}p_t\leq\sum_{t=1}^{k}p_t(b)=\sum_i P_i(b)$. Hence, by summing over all players and using the latter inequality we get the theorem.
\end{proof}

Now similar to the bandwidth allocation setting we can apply our simultaneous composability theorem even
under budgets when players have submodular valuations on the product lattice on $\N^m$. Submodularity
on $\N^m$ means that the functions must be concave coordinate-wise and must satisfy the decreasing differences.

\begin{corollary}[Simultaneous with Budgets]
If we run $m$ greedy multi-unit auctions and bidders have submodular valuations on $\N^m$ and budget constrains
then every CE and BNE achieves at least $\frac{e-1}{2e}\approx \frac{1}{3.16}$ of the expected optimal
effective welfare.
\end{corollary}

For sequential composition we require that the bidders are unit-demand over mechanisms: e.g. they have mechanism specific concave value functions and that their utility is the maximum over all mechanisms of the utility they get from each mechanism, $v_i(k_i) = \max_j v_{ij}(k_{ij})$. Observe that such valuations are a generalization of the
standard notion unit-demand valuations where players just want one unit. We could simulate unit-demand valuations
with unit-demand over mechanisms by just saying that $v_{ij}(k_{ij}) = \hat{v}_{ij}$ if $k_{ij}>=1$. Our notion
of unit-demand valuations over mechanisms just says that you should pick the mechanism that gave you
the maximum value for the units it gave you.

\begin{corollary}[Sequential]
If we run $m$ greedy multi-unit auctions sequentially and bidders have unit-demand valuations over mechanisms
then every CE and BNE achieves at least $\frac{e-1}{4e}\approx \frac{1}{6.32}$ of the expected optimal social welfare.
\end{corollary}

One could also think of running a second-price equivalent of Mechanism \ref{mech:first-price-multi-unit} which is
described in Mechanism \ref{mech:second-price-multi-unit}.

\renewcommand{\algorithmcfname}{MECHANISM}
\begin{algorithm}[h]\label{mech:second-price-multi-unit}
\SetKwInOut{Input}{Input}\SetKwInOut{Output}{Output}
\BlankLine
\nl  Solicit marginal bids $b_{i1},\ldots,b_{ik}$ from each player $i$ which are restricted to be decreasing\;
\nl  At each iteration allocate the extra unit to the bidder that has the maximum marginal
value for getting it conditional on the items he has already been allocated\;
\nl  For each unit that a player receives charge him the highest marginal bid that the mechanism didn't allocate to
\caption{Greedy Multi-Unit Threshold Price Auction with concave values.}
\end{algorithm}
\renewcommand{\algorithmcfname}{ALGORITHM}

Markakis and Telelis \cite{Markakis2012} studies exactly this auction and uses a no-overbidding assumption, where the
willingness-to-pay of a player is the sum of his $k$ highest marginal bids if he is allocated $k$ units. Under this no-overbidding
assumption and using similar analysis as in Theorem \ref{lem:greedy-multi-unit} we can prove that this auction
is weakly $\left(\frac{1}{2}\left(1-\frac{1}{e}\right),0,1\right)$-smooth, thereby implying an efficiency guarantee of $\frac{1}{4}\left(1-\frac{1}{e}\right)$. This largely improves upon the results of Markakis et al. \cite{Markakis2012} where only a logarithmic bound in
the number of units $O(\log(k))$ was proved for the case of mixed and Bayes-Nash equilibria. Our bound also has implications
for budgets and simultaneous and sequential composition.

\textbf{Uniform Price Auction.} Another multi-unit auction that has been widely used in practice is the Uniform Price Auction. In the uniform price auction every bidder is asked to report a pair $(q_i,b_i)$ where
$q_i\in \N$ is a quantity and $b_i$ is a  per-unit bid. The auction then orders the bids in decreasing order and serves the units until reaching capacity.

\renewcommand{\algorithmcfname}{MECHANISM}
\begin{algorithm}[h]\label{mech:uniform-price}
\SetKwInOut{Input}{Input}\SetKwInOut{Output}{Output}
\BlankLine
\nl  Solicit quantity, bid pairs $(q_i,b_i)$ from each player $i$\;
\nl  Let $Q_t$ be the total units allocated until iteration $t$\;
\nl  At each iteration $t$ pick  unallocated player with highest  $b_i$ and allocate him $\min\{q_i,k-Q_t\}$,
until all units are sold \;
\nl  Charge everyone the highest losing bid, i.e. the bid of the last player that was partially satisfied or if the last player was completely satisfied then the bid of highest player that was unallocated.
\caption{Uniform Price Auction with concave values.}
\end{algorithm}
\renewcommand{\algorithmcfname}{ALGORITHM}

Uniform price auctions are frequently used in practice because they have the advantage that no-matter
what the players bid, everyone pays the same price for the allocated items. Hence, they give a fairness feeling and also avoid any friction when someone was allocated the same unit at a different price.

In such  an auction the willingness-to-pay of a player that received $k_i$ units and bid $b_i$ per
unit is exactly $k_i b_i$ since in the worst case the highest losing bid could be just below your bid.

\begin{lemma}
The Uniform Price Auction is conservatively and weakly $$\left(\frac{1}{2}\left(1-\frac{1}{e}\right),0,1\right)$$ smooth
when bidders valuations $v_i:\N\rightarrow \R^+$ are concave with $v_i(0)=0$.
\end{lemma}
\begin{proof}
Consider a strategy profile $a=(k,b)$ and a bid profile $v:\N\rightarrow \R^+$. Suppose that bidder $i$ deviates to bidding $a_i'=(k_i^*,t)$ where $t$ is a number drawn randomly according to the distribution with probability density function
$f(t)=\frac{\beta}{\frac{v_i(k_i^*)}{k_i^*}-t}$ and support $[0,\frac{v_i(k_i^*)}{k_i^*}\left(1-\frac{1}{e^{1/\beta}}\right)]$.

Denote by $B_t$ the bid of the $t$-th last unit sold. Similar to Mechanism \ref{mech:first-price-multi-unit} it holds that if $t>B_j$ then the player is allocated at least $j$ units. Thereby using similar analysis as in the proof of Lemma \ref{lem:greedy-multi-unit} we can show that the above deviation yields utility at least:
\begin{align*}
u_i(a_i',a_{-i})\geq \frac{\beta}{2}\left(1-\frac{1}{e^{1/\beta}}\right)v_i(k_i^*)-\beta \sum_{t=1}^{k_i^*}B_t
\end{align*}
Then by summing among all players we can derive:
\begin{align*}
\sum_i u_i(a_i',a_{-i})\geq~& \frac{\beta}{2}\left(1-\frac{1}{e^{1/\beta}}\right)v_i(k_i^*)-\beta \sum_i k_i(a) b_i\\
=~& \frac{\beta}{2}\left(1-\frac{1}{e^{1/\beta}}\right)v_i(k_i^*)-\beta \sum_i B_i(a_i,k_i(a))
\end{align*}
Setting $\beta=1$ we get the theorem.
\end{proof}

We also get the same composability guarantees as Mechanism \ref{mech:first-price-multi-unit}:
\begin{corollary}[Simultaneous with Budgets]
If we run $m$ uniform price auctions and bidders have submodular valuations on $\N^m$ and budget constrains then every CE and BNE that satisfies the no-overbidding assumption achieves at least $\frac{e-1}{4e}\approx \frac{1}{6.32}$ of the expected optimal effective welfare.
\end{corollary}

\begin{corollary}[Sequential]
If we run $m$ uniform price auctions sequentially and bidders have unit-demand valuations over mechanisms
then every CE and BNE achieves at least $\frac{e-1}{4e}\approx \frac{1}{6.32}$ of the expected optimal social welfare.
\end{corollary}

\subsection{Combinatorial Auctions with XOS valuations}\label{SEC:GREEDY-MARGINAL}
Consider a combinatorial auction setting with $k$ items and $n$ bidders who have submodular
valuations. For such a setting Lehman et al. \cite{Lehmann2001} gave a greedy $2$-approximation algorithm. 
In this section we analyze a mechanism based on that greedy allocation. We will consider bidders
that have XOS valuations on sets of items. One problem with analyzing the greedy algorithm as a mechanism
is that it requires for the players to submit and commit to their whole valuation. That requires an exponential communication. Hence, we will consider here a simplification of the algorithm where the bidders are asked
to submit only additive proxies to their valuations.

Remember that XOS valuations when specialized on valuations defined on sets of items are characterized by the following property: for any set $S$ there exists an additive valuation $v_{ij}^S$ such that $v_i(S)=\sum_{j\in S} v_{ij}^S$ and such that for any other set $T$: $v_i(T)\geq \sum_{j\in T}v_{ij}^S$.

\renewcommand{\algorithmcfname}{MECHANISM}
\begin{algorithm}[h]\label{mech:greedy-XOS-first}
\SetKwInOut{Input}{Input}\SetKwInOut{Output}{Output}
\BlankLine
\nl  Solicit additive bids/valuations $\tilde{v}_{ij}$\;
\nl  Initialize everyone's allocation $S_i=\emptyset$\;
\nl  At each iteration pick an unallocated item $j$ and 
bidder $i$ that maximizes the marginal valuation: $v_{ij}$ among all such pairs and make the allocation.\;
\nl Repeat until all items are allocated.\;
\nl Charge each player his reported valuation for the set he won $\tilde{v}_i(S_i)$\;
\caption{Greedy Marginal Allocation Auction.}
\end{algorithm}
\renewcommand{\algorithmcfname}{ALGORITHM}
The best known approximation factor to the problem of welfare maximization for XOS bidders is $1-\frac{1}{e}$
and was given by Feige \cite{Feige2006}.

Here we observe that the above mechanism that uses the greedy allocation rule with the restriction that the players submit additive proxies is equivalent to running $m$ simultaneous first price auctions with XOS valuations and thereby is $(1-\frac{1}{e},1)$-smooth. Therefore the efficiency guarantees that it provides are 
the same as the best approximation algorithm for the optimization problem. 

Similarly, if instead of the first price we charged each player the minimum he had to bid to win each item, then that would be equivalent to simultaneous second-price auctions and therefore would be weakly $(1,0,1)$-smooth, leading to an efficiency guarantee of $1/2$.

%
%
%
%

\section{Relaxed Smoothness for\\ General Games} 
\label{app:general-smooth}
Consider a complete information normal form game that consists of a set of players $[n]$, a strategy space $S_i$ for each player $i\in [n]$ and a utility function $u_i:\times_i S_i\rightarrow \R^+$. We will consider a social welfare function $SW:\times_i S_i\rightarrow \R^+$ that captures the total utility of all the entities participating in the game $SW(s)=\sum_i u_i(s)$, possibly also considering welfare to participants that are not explicitly modeled as players, such the welfare of the customers served  or the revenue of the mechanism.

Given such a game we are interested in studying how bad the social welfare can be at outcomes that
correspond to well-established solution concepts in game theory and in comparison with the
social welfare optimal outcome. Throughout this section we will denote with $\opt$ the optimal social
welfare value. 
We will quantify the efficiency of a solution concept using the notion of Price of
Anarchy which is the ratio of the optimal social welfare over the worst expected equilibrium social
welfare (since an equilibrium might correspond to a distribution over outcomes).


\textbf{When is a game smooth?} Roughgarden \cite{Roughgarden2009} defined smoothness by the existence of a single socially aware strategy profile $s^*$ such that in any strategy profile $s$, an agent could close his eyes, forget what he and everyone else was doing previously and then play a socially-aware strategy $s_i^*$, and produce a good fraction of his share of the optimal outcome. More precisely, that such an effect is true in the aggregate.

\begin{defn}[Smooth Game, Roughgarden \cite{Roughgarden2009}]
A utility maximization game is $(\lambda,\mu)$-smooth if there exists a (possibly randomized) strategy $\st_i^*$ for each $i\in N$, such that for all strategy profiles $s \in \times_i S_i$:
\begin{equation}
\sum_i u_i(\st_i^*,s_{-i})\geq \lambda \opt - \mu \sum_i u_i(s)
\end{equation}
\end{defn}

The most important property about smooth games is that any Coarse Correlated Equilibrium
of a $(\lambda,\mu)$-smooth game has expected social welfare at least $\frac{\lambda}{1+\mu}\opt$,
i.e. the Price of Anarchy of Coarse Correlated Equilibria is at most $\frac{1+\mu}{\lambda}$ \cite{Roughgarden2009}.
Hence, the above states that if players are using no-regret learning strategies when playing
a smooth game, then in the limit the expected social welfare of the empirical distribution of
play will be approximately optimal.



Such a property though it holds for several natural classes of games might be too strong. A more natural class of deviations
would allow the deviating player to condition his deviation at least on
what he/she was doing previously.\footnote{One could also allow the deviating player to condition his deviation on what
everyone was doing prior to the deviation, i.e. on the whole strategy profile $s$, but such a property would imply a Price of Anarchy
bound only for the Pure Nash Equilibria of the game and nothing more, not for learning outcomes, not for Bayesian version of the game, not even for mixed equilibria of the full information game.}
This is exactly the notion of \textit{conditional smoothness}
that we introduce:

\begin{defn}[Conditionally Smooth Game]\label{def:cond-smooth}
A utility maximization game is conditionally $(\lambda,\mu)$-smooth if for any strategy
profile $s$ there exists a (possible randomized) strategy $\st_i^*(s_i)$ for all $i\in N$ such that:
\begin{equation}
\sum_i u_i(\st_i^*(s_i),s_{-i})\geq \lambda \opt - \mu \sum_i u_i(s)
\end{equation}
\end{defn}

However, the fact that we now allow players to condition their deviation on what they were doing
in the previous strategy profile comes at a small loss. Specifically, the Price of Anarchy of
Coarse Correlated Equilibria of a $(\lambda,\mu)$-conditionally smooth game could be much higher
than $\frac{1+\mu}{\lambda}$. However, we can still get that the good Price of Anarchy is
still achieved by all Correlated Equilibria of the game, and thereby it implies that when
players are using  No-Internal-Regret learning strategies then the expected social welfare
of the empirical distribution will be approximately optimal.

\begin{theorem}\label{thm:correlated_poa}
The expected social welfare of any {\bf Correlated Equilibrium} of a {\bf conditionally} $(\lambda,\mu)$-smooth
 game is at least $\frac{\lambda}{1+\mu}\opt$. If all players are using  No-Internal-Regret learning strategies then the
 expected social welfare  is at least $\frac{\lambda}{1+\mu}\opt$.
\end{theorem}
\begin{proof}
A correlated equilibrium is a distribution $\st$ over strategy profiles such that for every player $i$, for every strategy $s_i$ in the support of $\st$ and for every other strategy $s_i'$ in $S_i$, player $i$ doesn't benefit from switching to $s_i'$ whenever he was playing $s_i$.
Since, no player has profit of deviating to $\st_i^*(s_i)$ whenever he was playing $s_i$ we get:
\begin{equation}
\forall i \; \forall s_i: \E_{\st}[u_i(\st)|s_i]\geq \E_{\st}[u_i(\st_i^*(s_i),\st_{-i})|s_i]
\end{equation}
The theorem follows by taking expectation over $s_i$ and adding over all players:
\begin{align*}
\E_{\st}[SW(\st)]=~&\E_{\st}[\sum_i u_i(\st)]\geq \E_{\st}\left[\sum_i u_i(\st_i^*(s_i),\st_{-i})\right]\\
\geq~& \lambda \opt - \mu\E_{\st}\left[\sum_i u_i(\st)\right]\\
=~& \lambda\opt - \mu \E_{\st}[SW(\st)]
\end{align*}
\end{proof}

%
%
%
\subsection{Bayesian Games}
We consider the following model of Bayesian Games: Each player
has a type space $T_i$ and a probability distribution $D_i$ defined on $T_i$.
The distributions $D_i$ are independent, the types $T_i$ are disjoint and we denote with $D=\times_i D_i$.
Each player has a set of actions $A_i$ and let $A=\times_i A_i$. The utility of a player is a function
$u_i: T_i\times A\rightarrow \R^+$, which we will denote by $u_i^{t_i}(a)$. The strategy of each player is a function
$s_i: T_i\rightarrow A_i$. At times we will use the notation $s(t)=(s_i(t_i))_{i\in N}$
to denote the vector of actions given a type profile $t$ and $s_{-i}(t_{-i})=(s_j(t_j))_{j\neq i}$
to denote the vector of actions for all players except $i$.

This definition makes a few basic assumptions about the class of Bayesian Games we consider:
a players utility is affected by the other players' types only implicitly through their actions
and not directly from their types, the set of actions available to an agent doesn't depend on their types,
and  the players' types are distributed independently.

We will use the most dominant solution concept in
incomplete information games, the Bayes-Nash Equilibrium (BNE). Our results hold for
mixed Bayes-Nash Equilibria too, but for simplicity of presentation we are
going to focus on Bayes-Nash Equilibria in pure strategies. A Bayes-Nash
Equilibrium is a strategy profile such that each player maximizes his expected
utility conditional on his private information.
$$\forall a \in A_i:\E_{t_{-i}|t_i}[u_i^{t_i}(s(t))]\geq \E_{t_{-i}|t_i}[u_i^{t_i}(a,s_{-i}(t_{-i})]$$
Given an action profile $a$ and a type vector $t$ the social welfare of the game $SW^t(s)$ is at least the sum of the player's utilities $\sum_i u_i^{t_i}(a)$. The welfare of a strategy profile $s$ is the expected welfare over the types
$$SW(s)=\E_t[SW^t(s(t))]=\E_t\left[\sum_i u_i^{t_i}(s(t))\right]$$
Given a type profile $t$ we denote with $\opt(t)$ the optimal social welfare for type
profile $t$.

As our measure of inefficiency we will use the \textit{Bayes-Nash Price of Anarchy} which
is defined as the ratio of the expected optimal social welfare over the expected
social welfare achieved at the worst Bayes-Nash Equilibrium:
$$\sup_{s \text{ is } BNE}\frac{\E_{t}[\opt(t)]}{\E_t[SW(s(t))]}$$

\subsection{Extension Theorem to Bayesian Setting}
In this section we show that the inefficiency bound implied by the weaker notion of
conditional smoothness that we introduced in this section also carries over to Bayesian versions of a game, i.e. if a Bayesian game is conditionally $(\lambda,\mu)$-smooth
for each complete information game that corresponds to each instance of the type profile then
the Bayes-Nash Price of Anarchy of the incomplete information game is also $\lambda/(1+\mu)$.

A Bayesian Game is conditionally-smooth if each induced complete information game is conditionally smooth in the sense of \ref{def:cond-smooth}.

\begin{defn}
A Bayesian game is $(\lambda,\mu)$-conditionally smooth if for any $t\in T$ and for any action profile $a\in A$ there exists a (possible randomized) action $\al_i^*(t,a_i)$ for each $i\in N$ such that:
\begin{equation}
\sum_i u_i^{t_i}(\al_i^*(t,a_i),a_{-i})\geq \lambda \opt(t) - \mu\sum_i u_i^{t_i}(a)
\end{equation}
\end{defn}

Next we prove our extension theorem, that if a full information game is conditionally-smooth, than the corresponding Bayesian game also has low price of anarchy. Note that the function $\al^*(t,s)$ allows the coordinate of player $i$ to depends on the whole type profile $t$, and not only the type of player $i$, as a result player $i$'s coordinate $\al_i^*(t,a_i)$ cannot be directly used as deviation for that player in the Bayesian Game. Previous papers by Roughgarden \cite{Roughgarden2012} and Syrgkanis \cite{Syrgkanis2012} managed to get around this problem by using a random sampling technique. However, those papers use the stronger definition of smoothness where the deviating strategy doesn't depend
on the previous action of the deviating player. Here we extend previous results for our weaker conditional smoothness property.

\begin{theorem}\label{thm:bayesian_poa}
If a Bayesian Game is $(\lambda,\mu)$-conditionally smooth then it has Bayes-Nash Price of Anarchy
at most $\lambda/(1+\mu)$.
\end{theorem}
\begin{proof}
Let $t$ be a type profile and $s(t)$ be a Bayes-Nash Equilibrium strategy profile that corresponds to type $t$.
Let $\al^*(t,a)=(\al_i^*(t,a_i))_{i\in N}$ be the action profile designated by the smoothness property of the game for a type profile $t$ and an action profile $a$.

We will consider the following randomized deviation for each player $i$ that depends only on the information that he has which is only his own type $t_i$: He random samples a strategy profile $w\sim F$. Then he plays $\al_i^*((t_i,w_{-i}),s_i(w_i))$. That is, the player considers the equilibrium strategies  $s(w)$, using the randomly sampled type (including the random sample of his own type), and deviates from this strategy profile using the strategy $\al^*$ given by the smoothness property using his true type $t_i$.  Using the strategy $s(w)$ as the base, corresponds to a bluffing technique that was introduced in \cite{Syrgkanis2012} in the context of sequential first price auctions, where player $i$ ``pretends'' that his type is $w_i$.

Since this is not a profitable deviation for player $i$, it means:
\begin{align*}
\E_{t}[u_i^{t_i}(s(t))]\geq~& \E_t \E_w [u_i^{t_i}(\al_i^*((t_i,w_{-i}),s_i(w_i)),s_{-i}(t_{-i})]\\
=~& \E_t \E_w[u_i^{w_i}(\al_i^*((w_i,w_{-i}),s_i(t_i)),s_{-i}(t_{-i}))]\\
=~& \E_t \E_w[u_i^{w_i}(\al_i^*(w,s_i(t_i)),s_{-i}(t_{-i}))]
\end{align*}
Summing over all players and using the conditional smoothness property we get:
\begin{align}
\E_t[\sum_i u_i^{t_i}(s(t))]\geq& \E_t \E_w[\sum_i u_i^{w_i}(\al_i^*(w,s_i(t_i)),s_{-i}(t_{-i}))]\nonumber\\
\geq& \E_t \E_w[\lambda \opt(w) - \mu\sum_i u_i^{w_i}(s(t))]\nonumber\\
=& \lambda \E_w[\opt(w)] - \mu \sum_i \E_t\E_w[u_i^{w_i}(s(t))]\label{eqn:bayes-nash}
\end{align}
If in the line before the last one we had $\sum_i u_i^{t_i}(s(t))$ then we would directly get our result. However,
the fact that there is this misalignment between the player types and the strategy evaluated, we need more work. In fact we are going to prove that:
\begin{equation}\label{eqn:misalignment}
 \E_{t} \E_{w}[\sum_i u_i^{w_i}(s(t))] \leq \E_t[\sum_i u_i^{t_i}(s(t))]
\end{equation}
To achieve this we are going to use again the equilibrium definition: no player $i$ of some type $w_i$
wants to deviate to playing as if he was some other type $t_i$, which gives us:
\begin{align*}
 \E_{t_{-i}}[u_i^{w_i}(s_i(t_i),s_{-i}(t_{-i})] \leq \E_{t_{-i}}[u_i^{w_i}(s_i(w_i),s_{-i}(t_{-i})]
\end{align*}
Taking expectation over $t_i$ and $w_i$ we have:
\begin{align*}
 \E_w\E_t[u_i^{w_i}(s_i(t_i),s_{-i}(t_{-i})] ~\leq~ & \E_{w_i}\E_{t_i}\E_{t_{-i}}[u_i^{w_i}(s_i(w_i),s_{-i}(t_{-i})]\\
 ~=~& \E_{w_i}\E_{t_{-i}}[u_i^{w_i}(s_i(w_i),s_{-i}(t_{-i})] \\
 ~=~& \E_{t_i}\E_{t_{-i}}[u_i^{t_i}(s_i(t_i),s_{-i}(t_{-i})] \\
~=~& \E_t[u_i^{t_i}(s(t))]
\end{align*}
where the equation starting the last line is just a change of variable names. Summing over all players gives us inequality (\ref{eqn:misalignment}). Now combining inequality (\ref{eqn:misalignment})
with inequality (\ref{eqn:bayes-nash}) we get:
\begin{equation}
 \E_t[SW^t(s(t))]\geq \lambda \E_w[\opt(w)] - \mu \E_t[SW^t(s(t))]
\end{equation}
which gives the theorem.
\end{proof}

\section{Omitted Proofs}
\subsection{Section \ref{SEC:VALUATIONS}: Hierarchy of Valuations}

\begin{proofof}{Theorem \ref{thm:xos_equivalence}}
We will give the proof of the general version of the theorem for $\beta$-fractionally subadditive valuation. Let $\X=\times_j \X_j$ denote the space of outcomes for all mechanisms. We first show that if the valuation is \textit{$\beta$-XOS over outcomes of mechanisms}
then it is also \textit{$\beta$-fractionally subadditive over outcomes of mechanisms}. Suppose that there exists a set of additive valuations $\mathcal{L}$ such that: $\sup_{\ell \in \El}\sum_{j\in [m]}v_j^\ell(x_j)\leq v(x)\leq \beta \sup_{\ell \in \El}\sum_{j\in[m]}v_j^\ell(x_j)$. Now consider an allocation $x^*$ and a fractional cover $(a_x)_{x\in X}$ of $x^*$, i.e. for all $j\in[m]$, $ \sum_{x:x_j=x_j^*}a_x\geq 1$. Now we have:
\begin{align*}
\sum_{x\in X}a_x v(x)\geq~& \sum_{x\in X}a_x \sup_{\ell\in \mathcal{L}}\sum_{j\in [m]}v_j^{\ell}(x_j)\\
\geq ~& \sup_{\ell\in \mathcal{L}}\sum_{x\in X}a_x\sum_{j\in [m]}v_j^{\ell}(x_j)\\
=~&\sup_{\ell\in \mathcal{L}}\sum_{j\in [m]}\sum_{x\in X}a_x v_j^{\ell}(x_j)\\
\geq~&\sup_{\ell\in \mathcal{L}}\sum_{j\in [m]}v_j^{\ell}(x_j^*) \sum_{x\in X: x_j=x_j^*}a_x \\
\geq~&\sup_{\ell\in \mathcal{L}}\sum_{j\in [m]}v_j^{\ell}(x_j^*) \geq \frac{1}{\beta} v(x^*)
\end{align*}
Thus the valuation is also $\beta$-fractionally subadditive.

Now we prove the opposite direction: if a valuation is \textit{$\beta$-fractionally subadditive over outcomes of mechanisms} then it is also \textit{$\beta$-XOS over outcomes of mechanisms}. Consider the following linear program associated with an outcome $x^*\in\X$:
\begin{align}
V(x^*)=~&\min_{a_x} \sum_{x\in X}a_x v(x)\nonumber\\
\text{ s.t. } & \sum_{x: x_j=x_j^*} a_x \geq 1 \text{~~~for all } j\in [m]\label{eqn:primal-lp}\\
& a_x\geq 0 \text{~~~for all } x\in \X\nonumber
\end{align}
By the property of $\beta$-fractionally subaddtive valuations, since the set of feasible solutions to the above linear program, constitutes a fractional
cover of $x^*$ we know that $\beta V(x^*)\geq v(x^*)$. In addition we know that we can achieve $v(x^*)$ by just setting $a_{x^*}=1$ and $a_x=0$ for any other $x\in \X$. Hence, $V(x^*)\leq v(x^*)\leq \beta V(x^*)$.

Now consider the dual of the above linear program:
\begin{align}
C(x^*)=~&\max_{t_j} \sum_{j \in [m]}t_j\nonumber\\
\text{ s.t. } & \sum_{j: x_j=x_j^*} t_j \leq v(x) \text{~~~for all } x\in \X \label{eqn:dual-lp}\\
& t_j\geq 0 \text{~~~for all } j\in [m]\nonumber
\end{align}

By LP duality we know that $V(x^*)=C(x^*)$. Let $t_j^{x^*}$ be an optimal solution to the dual associated with allocation $x^*$.
Now consider the following additive valuation: $v_j^{x^*}(x_j)=t_j$ if $x_j=x_j^*$ and $0$ otherwise. By the constraints of the dual
we know that $v(x)\geq \sum_{j: x_j=x_j^*} t_j= \sum_{j\in [m]}v_j^{x^*}(x_j)$. Therefore:
$$v(x)\geq \sup_{x^*\in X}\sum_{j\in [m]}v_j^{x^*}(x_j)$$
In addition by LP duality we know:
\begin{align*}
 \sup_{x^*\in \X}\sum_{j\in [m]}v_j^{x^*}(x_j)\geq~& \sum_{j\in [m]}v_j^{x}(x_j)
=\sum_{j\in[m]}t_j^{x}=C(x)\\=~&V(x)\geq \frac{1}{\beta}v(x)
\end{align*}
Therefore we get that:
$$\sup_{x^*\in X}\sum_{j\in [m]}v_j^{x^*}(x_j)\leq v(x)\leq \beta \sup_{x^*\in X}\sum_{j\in [m]}v_j^{x^*}(x_j)$$
Hence, the valuation is also $\beta$-XOS.
\end{proofof}

\vspace{.2in} \textbf{Remark.} Observe that in the theorem above we used single-minded component valuations $v_j^\ell:\X_j\rightarrow \R_+$ of the form: $v_j^{\hat{x}}(x_j)=c$ if $x_j=\hat{x}_j$ and $0$
otherwise. Thus this gives the corollary:
\begin{corollary}
Any fractionally subadditive valuation can be expressed as an XOS valuation using only single-minded induced valuations.
\end{corollary}

\begin{proofof}{Theorem \ref{thm:submod-xos}}
We will prove that there exist a set of additive valuations $\mathcal{L}$ such that
$$\forall x: v(x)=\sup_{\ell \in \mathcal{L}}\sum_{j\in [m]} v_j^{\ell}(x_j)$$
Each additive valuation in $\mathcal{L}$ will be associated with an outcome $x$. Denote with $M_j=\{j'\in [m]:j'\leq j\}$. The additive valuation associated with an outcome $x$ will then be:
\begin{equation}
v_j^x(\tilde{x}_j)=\begin{cases}
v(x_{M_j})- v(x_{M_{j-1}})     & \text{ if } \tilde{x}=x_j\\
0 & \text{ o.w. }
                      \end{cases}
\end{equation}
First observe that:
\begin{equation*}
 \sum_j v_j^x(x_j)=\sum_j v(x_{M_j})- v(x_{M_{j-1}})=v(x)
\end{equation*}
Next we will show that for all $\tilde{x},x\in \X: v(\tilde{x})\geq \sum_j v_j^x(\tilde{x}_j)$. The latter two facts together will establish that for all $\tilde{x}\in \X: v(\tilde{x})=\sup_{x} \sum_j v_j^x(\tilde{x}_j)$ which will complete theorem.

Fix two outcome vectors $\tilde{x}, x$. Let $S=\{j\in [m]:\tilde{x}_j = x_j\}$ and $S_j=\{j'\in S:j'\leq j\}=M_j\cap S$. By the set-monotonicity of the valuation we have:
\begin{align*}
 v(\tilde{x}) = v(\tilde{x}_S,\tilde{x}_{-S}) = v(x_S,\tilde{x}_{-S})\geq v(x_{S})
\end{align*}
Now observe that:
\begin{align*}
 v(\tilde{x}) \geq v(x_{S}) = \sum_{j\in S} v(x_{S_j})-v(x_{S_{j-1}})
\end{align*}
Since for all $j\in S: S_j\subseteq M_j$, by set-submodularity we get that:
\begin{align*}
 v(\tilde{x})&~ \geq v(x_{S}) = \sum_{j\in S} v(x_{S_j})-v(x_{S_{j-1}})\\
&~\geq \sum_{j\in S} v(x_{M_j})-v(x_{M_{j-1}})\\
&~= \sum_{j\in S} v_j^x(\tilde{x}_j) = \sum_{j\in [m]} v_j^x(\tilde{x}_j)
\end{align*}
The latter equality follows from the fact that for all $j\notin S: v_j^x(\tilde{x}_j)=0$, by definition.
\end{proofof}

\begin{proofof}{Theorem \ref{thm:subadd-xos}}
This proof is the generalization of the analogous proof for the case of valuations
defined on sets, presented in \cite{Bhawalkar2011}.

We will show that there exists a set of additive valuations $\El$ such that:
$$\sup_{\ell\in \El} \sum_j v_j^\ell(x_j)\leq v(x) \leq H_m \sup_{\ell \in \El} \sum_j v_j^\ell(x_j)$$
Each additive valuation in $\El$ will be associated with an outcome $x\in \X$. We define the additive valuation associated with outcome $x$ using the iterative process presented
in Algorithm \ref{alg:subadditive-computation}.

\begin{algorithm}[h]
\SetKwInOut{Input}{Input}\SetKwInOut{Output}{Output}
\Input{An outcome $x\in \X$ and a valuation $v: \X\rightarrow \R$}
\Output{Monotone valuations $v_j^x:\X_j\rightarrow \R$ for each $j\in [m]$}
\BlankLine
\nl  Set $C=\emptyset$\;
\nl \While{$C\neq [m]$}{
    Pick $A=\arg\min_{A'\subseteq [m]} \frac{v(x_{A'})}{|A'-C|}$\;
    \For{each $j\in A-C$}{
      \begin{equation}
	v_j^x(\tilde{x}_j)=\begin{cases}
		    \frac{v(x_A)}{|A-C|H_m} & \text{ if } \tilde{x}= x_j\\
		    0 & \text{ o.w. }
	           \end{cases}
      \end{equation}
    $C=C\cup A$\;}}
\caption{Procedure for computing the additive valuation associated with each outcome $x\in \X$.}
\label{alg:subadditive-computation}
\end{algorithm}

First we argue that for each $x,\tilde{x}\in X$: $v(\tilde{x})\geq \sum_j v_j^x(\tilde{x}_j)$.
Let $S=\{j\in [m]: \tilde{x}_j= x_j\}$. By the set-monotonicity of the valuation we have:
$v(\tilde{x})\geq v(x_{S})$. Hence, it suffices to show that:
$$v(x_{S})\geq \sum_j v_j^x(\tilde{x}_j)=\sum_{j\in S}v_j^x(\tilde{x}_j)$$
Consider the iteration $t$ of Algorithm \ref{alg:subadditive-computation} at which the $k$-th
element of $S$ is added in $C$. Since at that iteration the algorithm chose $A_t$ we
have:
$$v_k^x(\tilde{x}_j)=\frac{v(x_{A_t})}{|A_t-C|H_m}\leq \frac{v(x_S)}{(|S|-k+1)H_m}$$ Therefore:
\begin{align*}
 \sum_{j \in S}v_j^x(\tilde{x}_j)=\sum_{k=1}^{|S|}v_k^x(\tilde{x}_j)
\leq \frac{v(x_S)}{H_m}\sum_{k=1}^{|S|}\frac{1}{|S|-k+1}
= v(x_S)
\end{align*}
Hence, for all $\tilde{x}\in \X: v(\tilde{x})\geq \sup_{x\in \X}\sum_j v_j^x(\tilde{x}_j)$.

Now we show that for all $\tilde{x}\in X: v(\tilde{x})\leq H_m \sup_{x\in \X}\sum_j v_j^x(\tilde{x}_j)$.
Suppose that the algorithm takes $T$ iterations to complete the computation and that $A_1,\ldots,A_T$ are the sets picked at each iteration. Then:
\begin{align*}
 \sup_{x\in \X}\sum_j v_j^x(\tilde{x}_j)
\geq~& \sum_j v_j^{\tilde{x}}(\tilde{x}_j)
= \sum_{t=1}^{T} \sum_{j\in A_t}v_j^{\tilde{x}}(\tilde{x}_j)\\
=~&\sum_{t=1}^{T}\frac{v(\tilde{x}_{A_t})}{H_m}\geq \frac{v(\tilde{x}_{\cup_{t=1}^{T}A_t})}{H_m}=
\frac{v(\tilde{x})}{H_m}
\end{align*}
\end{proofof}

\begin{proofof}{Theorem \ref{thm:monotone-xos}}
Consider a valuation $v$ that is \textit{$\beta$-fractionally subadditive over outcomes of mechanisms} and monotone
with respect to a Cartesian ordering $(\succeq_{j})_{j\in [m]}$. We will prove that it is also \textit{$\beta$-XOS over outcomes of mechanisms} and such that each induced valuation $v_{j}^\ell:\X_{j}\rightarrow \R^+$, used
in the XOS representation, is monotone with respect to $\succeq_{j}$.

Consider the following variation of the linear program \eqref{eqn:primal-lp} used
in the proof of Theorem \ref{thm:xos_equivalence} associated with an outcome $x^*\in\X$:
\begin{align*}
V(x^*)=~&\min_{a_x} \sum_{x\in \X}a_x v(x)\\
\text{ s.t. } & \sum_{x: x_j\succeq_{j} x_j^*} a_x \geq 1 \text{~~~for all } j\in [m]\\
& a_x\geq 0 \text{~~~for all } x\in \X
\end{align*}
Observe that in this variation the first set of constraints is altered to include a summation over outcomes greater than or equal to $x_j^*$ and not only on outcomes equal
to $x_j^*$ as in LP \eqref{thm:xos_equivalence}.

Consider a feasible solution $a_x$ to  the above linear program. For $x\in\X$, let $S(x)=\{j\in [m]:x_j\succeq x_j^*\}$. By the monotonicity of the valuation we know that:
\begin{align*}
\sum_{x\in \X}a_x v(x) \geq \sum_{x\in \X}a_x v(x^*_{S(x)},x_{-S(x)}) = \sum_{x\in \X} \tilde{a}_x v(x)
\end{align*}
where, it is easy to see that:
\begin{align*}
\sum_{x: x_j = x_j^*} \tilde{a}_x = \sum_{x: x_j\succeq_{ij} x_j^*} a_x \geq 1
\end{align*}
where the last inequality follows from the constraints of the linear program.

Therefore $\tilde{a}_x$ is a fractional cover of $x^*$. Hence, by the property of $\beta$-fractionally subaddtive valuations we know that
$$\beta \sum_{x\in \X}a_x v(x^*_{S(x)},x_{-S(x)})\geq v(x^*)$$
Since, this holds for any feasible $a_x$ we get that $\beta V(x^*)\geq v(x^*)$. In addition we know that we can achieve $v(x^*)$ by just setting $a_{x^*}=1$ and $a_x=0$ for any other $x\in \X$. Hence, $V(x^*)\leq v(x^*)\leq \beta V(x^*)$.

Now consider the dual of the above linear program:
\begin{align*}
C(x^*)=~&\max_{t_j} \sum_{j \in [m]}t_j\\
\text{ s.t. } & \sum_{j: x_j\succeq_{j} x_j^*} t_j \leq v(x) \text{~~~for all } x\in \X\\
& t_j\geq 0 \text{~~~for all } j\in [m]
\end{align*}

By LP duality we know that $V(x^*)=C(x^*)$. Let $t_j^{x^*}$ be an optimal solution to the dual associated with allocation $x^*$.

Now consider the following induced valuations: $v_j^{x^*}(x_j)=t_j$ if $x_j\succeq_{j} x_j^*$ and $0$ otherwise.
By the constraints of the dual we know that $v(x)\geq \sum_{j: x_j\succeq_{j} x_j^*} t_j= \sum_{j\in [m]}v_j^{x^*}(x_j)$. Therefore:
$$v(x)\geq \sup_{x^*\in X}\sum_{j\in [m]}v_j^{x^*}(x_j)$$
In addition by LP duality we know:
\begin{align*}
 \sup_{x^*\in \X}\sum_{j\in [m]}v_j^{x^*}(x_j)\geq~& \sum_{j\in [m]}v_j^{x}(x_j)
=\sum_{j\in[m]}t_j^{x}=C(x)\\=~&V(x)\geq \frac{1}{\beta}v(x)
\end{align*}
Therefore we get that:
$$\sup_{x^*\in X}\sum_{j\in [m]}v_j^{x^*}(x_j)\leq v(x)\leq \beta \sup_{x^*\in X}\sum_{j\in [m]}v_j^{x^*}(x_j)$$
Hence, the valuation is also $\beta$-XOS.

To complete the theorem we just need to prove that the valuations $v_j^x(\cdot)$ are monotone under
the partial order $\succeq_{j}$, for each $x\in \mathcal{X}$ and for each $j\in[m]$. Observe that $v_j^x(\cdot)$ takes the same value (namely $t_j^x$) for all $\tilde{x}_j\succeq_j x_j$ and is $0$ for any other $\tilde{x}_j$.

Consider two outcomes $\tilde{x}_j, \hat{x}_j\in \X_j$, such that $\tilde{x}_j\succeq_j \hat{x}_j$.
If $\hat{x}_j\succeq_j x_j$ then by transitivity of $\succeq_j$ we also have that $\tilde{x}_j\succeq_j x_j$ and therefore $v_j^x(\tilde{x}_j)=v_j^x(\hat{x}_j)$. Otherwise, by definition we have $v_j^x(\hat{x}_j)=0$
and therefore trivially $v_j^x(\tilde{x}_j)\geq v_j^x(\hat{x}_j)$.
\end{proofof}

\begin{defn}[Lattice-Submodular]
If each poset $(\X_{j},\succeq_{j})$ forms a lattice then a valuation is lattice-submodular
if and only if it is submodular on the product lattice of outcomes:
\begin{equation}
\forall x,\tilde{x}\in \X:
v(x\vee \tilde{x})+v(x \wedge \tilde{x})\leq v(x)+v(\tilde{x})
\end{equation}
and it satisfies the diminishing marginal returns property iff:
\begin{equation}
\forall z\succeq y \in \X \implies \forall t\in \X: v(t\vee y)-v(y)\geq v(t\vee z)-v(z)
\end{equation}
\end{defn}
\begin{lemma}\label{lem:lattice-sub-dim}
 If a valuation satisfies the diminishing marginal property with respect to a lattice structure then it is also lattice-submodular. If the lattice is distributive and the valuation is monotone then the inverse also holds.
\end{lemma}
\begin{proof}
  Consider two outcomes $x,\tilde{x}\in \X$. Since $\tilde{x}\succeq x\wedge \tilde{x}$, by the
diminishing marginal returns property we have:
\begin{align*}
 v(x\vee \tilde{x})-v(\tilde{x})\leq~& v(x\vee (x\wedge\tilde{x}))-v(x\wedge \tilde{x})\\
=~&v(x)-v(x\wedge \tilde{x})
\end{align*}
By rearranging we get that:
\begin{align*}
 v(x\vee \tilde{x})+v(x\wedge \tilde{x})\leq v(x)+v(\tilde{x})
\end{align*}
Since, $x,\tilde{x}$ where arbitrary the latter holds for any pair and therefore the valuation is
submodular over outcomes of mechanisms.

Now consider a valuation that is monotone and submodular over outcomes and in addition the lattice
$(\X,\succeq)$ is distributive. We will show that for any $z\succeq y$ and for any $t\in X$ the
diminishing marginal property holds.
Invoke the submodular property for $x=t\vee y$ and $\tilde{x}=z$:
\begin{equation*}
 v((t\vee y)\vee z)+v((t\vee y)\wedge z)\leq v(t\vee y)+v(z)
\end{equation*}
Since $z\succeq y$ we have $t\vee y\vee z=t\vee z$. In addition, by distributivity of the lattice:
$(t\vee y)\wedge z=(t\wedge z)\vee (y\wedge z)=(t\wedge z)\vee y$. Thus:
\begin{equation*}
 v(t\vee z)+v((t\wedge z)\vee y)\leq v(t\vee y)+v(z)
\end{equation*}
Now by monotonicity of the valuation we know that
$$v((t\wedge z)\vee y)\geq v(y)$$ Thus we get:
\begin{equation*}
 v(t\vee z)+v(y)\leq v(t\vee z)+v((t\wedge z)\vee y)\leq v(t\vee y)+v(z)
\end{equation*}
By rearranging we get the diminishing marginal property:
\begin{equation*}
 v(t\vee z)-v(z)\leq v(t\vee y)+v(y)
\end{equation*}
\end{proof}

\begin{proofof}{Theorem \ref{thm:lattice-xos}}
Now we turn to the case when $(\X_j,\succeq_j)$ forms a lattice and the valuation is monotone and satisfies the diminishing marginal returns property. First observe that the product poset $(\X,\succeq)$ will also form a lattice and the meet and join functions are also well defined. In this case we will modify the definition of $v_j^x(\tilde{x}_j)$ used in Theorems \ref{thm:submod-xos}, \ref{thm:monotone-xos} as follows:
\begin{equation}
 v_j^x(\tilde{x}_j)=v(\tilde{x}_j\wedge x_j,x_{M_{j-1}},\bot_{-M_j})-v(x_{M_{j-1}})
\end{equation}
We first show that the set of additive valuations satisfy the XOS definition. Fix two outcomes
$x,\tilde{x}\in \X$ and let $\hat{x}=\tilde{x}\wedge x$. By the monotonicity of the valuations we have that:
\begin{equation}
 v(\tilde{x})\geq v(\hat{x})
\end{equation}
Now by the diminishing marginal returns property of the function over the product lattice and the fact that for all $j$, $\hat{x}_j\succeq_j x_j$ we have:
\begin{align*}
 v(\tilde{x})\geq~& v(\hat{x})= \sum_j v(\hat{x}_j,\hat{x}_{M_{j-1}},\bot_{-M_j})-v(\hat{x}_{M_{j-1}})\\
\geq~& \sum_j v(\hat{x}_j,x_{M_{j-1}},\bot_{-M_j})-v(x_{M_{j-1}})=\sum_j v_j^x(\tilde{x}_j)
\end{align*}

Now we show that each $v_j^x(\cdot)$ satisfies the diminishing marginal returns with respect to the lattice $(\X_j\succeq_j)$. Observe that the negative part in the definition of $v_j^x(\tilde{x}_j)$ is independent of $\tilde{x}_j$. Thus it suffices to show that the first part satisfies the diminishing
marginal returns. For that it suffices to show that the following function:
\begin{equation}
 v_j(\tilde{x}_j)=v(\tilde{x}_j\wedge x_j,x_{-j})
\end{equation}
satisfies the diminishing marginal returns as a function of $\tilde{x}_j$ for any $x\in \X$, whenever $v(\cdot)$ satisfies the diminishing marginal returns with respect to the product lattice. Since,
we assumed that the valuation is monotone and the lattices are distributive we will equivalently
show that $v_j(\cdot)$ is lattice-submodular whenever $v(\cdot)$ is:
\begin{gather*}
 v_j(y_j\wedge z_j)+v_j(y_j\vee z_j)= \\
v(y_j\wedge z_j\wedge x_j,x_{-j})+v((y_j\vee z_j)\wedge x_j,x_{-j})= \\
v((y_j\wedge x_j)\wedge (z_j\wedge x_j),x_{-j})+v((y_j\wedge x_j)\vee (z_j\wedge x_j),x_{-j})\geq \\
v(y_j\wedge x_j,x_{-j})+v(z_j\wedge x_j,x_{-j})=\\
v_j(y_j)+v_j(z_j)
\end{gather*}
Where the second equality follows from the distributivity of the lattice and the inequality follows
from the submodularity of $v(\cdot)$.
\end{proofof}

\subsection{Section \ref{SEC:SMOOTHNESS}:
Smooth Mechanisms}

\begin{proofof}{Theorem \ref{thm:smooth}}
A correlated equilibrium is a distribution $\al$ over action profiles $a\in \A$ such for every player $i$ and every strategy $a_i$ in the support of $\al$ and every $a_i'\in \A_i$, player $i$ doesn't benefit from switching to $a_i'$ whenever he was playing $a_i$. Since, no player has profit of deviating to $\al_i^*(v,a_i)$ we get that for any $a_i$ in the support of $\al$:
\begin{equation}
\E_{\al_{-i}|a_i}[u_i(a_i,\al_{-i})]\geq \E_{\al_{-i}|a_i}[u_i(\al_i^*(v,a_i),\al_{-i})]
\end{equation}
The theorem follows by taking expectation over $a_i$ and adding over all players:
\begin{align*}
\E_{\al}[\sum_i u_i(\al)]\geq~& \E_{\al}[\sum_i u_i(\al_i^*(v,\al_i),\al_{-i})]\\
\geq~& \lambda \opt(v) - \mu\E_{\al}[\sum_i P_i(\al)]
\end{align*}
By the quasi-linear  utilities we have that $v_i(X_i(a))=u_i(a)+P_i(a)$, hence
\begin{align*}
\E_{\al}[\sum_i v_i(X_i(\al))]\geq~& \lambda \opt(v) +(1-\mu)\E_{\al}[\sum_i P_i(\al)].
\end{align*}
The result now follows if $\mu \le 1$. When $\mu >1$, we use the fact that players have the possibility to withdraw from the mechanism and get 0 utility, we have that $\E_{\al}[v_i(X_i(\al))]\ge \E_{\al}[P_i(\al)]$, and we get the result.
\end{proofof}

\subsection{Section \ref{SEC:COMPOSING}: Composition Theorems}

\begin{proofof}{Theorem \ref{thm:sequentially}}
Consider a valuation profile $v$ and an action profile $\ac$ of the sequential composition. Remember 
that in the sequential composition $\aci$ is not a strategy $\acij$ for each $j$ but rather a 
whole contingency plan of what action $\acij(h_i^j)$ to use at mechanism $\Mj$, conditional on the observed history of play by player $i$ up till mechanism $\Mj$. 

Let $x^*$ be the optimal allocation for valuation profile $v$. As stated, we assume that players have unit-demand valuations of the form: 
\begin{equation*}
v_i(\xsi) = \max_{j\in [m]} v_i^j(\xsij)
\end{equation*}
where $v_i^j \in \Vij$. We will denote with $j_i^*=\arg\max_{j\in [m]}v_i^j(x_{ij}^*)$, i.e.  
$v_i(x_i^*)=v_i^{j^*_i}(x_{ij_i^*}^*)$.

To prove the theorem we will give a randomized deviation $\al_i^*(v,a_i)$ for each agent $i$ such that:
\begin{align*}
\sum_i u_i^{v_i}(\al_i^*(v,a_i),a_{-i})\geq \lambda \sum_i v_i(x_i^*)-(1+\mu)\sum_i P_i(a)
\end{align*}
Remember that this will be a randomization over contingency plans. 

Consider the following type of randomized deviation $\al_i^*=\al_i^*(v,a_i)$ for player $i$ (the deviation will also be a contingency plan for each history of play): he plays exactly as in $a_i$ until mechanism $j=j^*_i$ and then he plays some randomized action $\al_{ij}^*$ that
will be determined later on and will be related to the smoothness of mechanism $\Mj$. The utility of player $i$ from this deviation is at least:
\begin{align*}
u_i^{v_i}(\al_i^*,a_{-i})\geq~& E_{\al_{ij}^*}[v_i^j(X_i^j(\al_{ij}^*,a_{-i}^j(h_{-i}^j)))-P_i^j(\al_{ij}^*,a_{-i}^j(h_{-i}^j))]\\&-P_i^{j-}(a)\\
\geq~& E_{\al_{ij}^*}[v_i^j(X_i^j(\al_{ij}^*,a_{-i}^j(h_{-i}^j)))-P_i^j(\al_{ij}^*,a_{-i}^j(h_{-i}^j))]\\&-P_i(a)
\end{align*}
where $P_i^{j-}(a)$ is the total payment that player $i$ made to mechanisms that happened prior to $\Mj$. Also $a_{-i}^j(h_{-i}^j)$ is the action profile submitted by the rest of the players at mechanism $j$ when players use contingency plan $a$ in the global game and hence each observes a history $h_i^j$ produced by this plan.

Summing over all players we get:
\begin{align*}
 \sum_i u_i^{v_i}(\al_i^*,a_{-i})\geq~~~~~~~~~~~~~~~~~~~~~~~~~~~~~~~~~~~~~~~~~~~~~~~~~~~~~~~~\\
\sum_{j}\sum_{i: j^*_i=j}  E_{\al_{ij}^*}[v_i^j(X_i^j(\al_{ij}^*,a_{-i}^j(h_{-i}^j)))-P_i^j(\al_{ij}^*,a_{-i}^j(h_{-i}^j))]\\-\sum_i P_i(a)
\end{align*}

Now observe that
$$\sum_{i: j^*_i=j}  E_{\al_{ij}^*}[v_i^j(X_i^j(\al_{ij}^*,a_{-i}^j(h_{-i}^j)))-P_i^j(\al_{ij}^*,a_{-i}^j(h_{-i}^j))]$$
represents the sum of utilities when each player unilaterally deviates to $\al_{ij}^*$ in mechanism $\Mj$ while previously everyone was playing $a_i^j(h_i^j)$ (the action that they are submitting to mechanism $\Mj$ when each seeing a history of play $h_i^j$ in previous mechanisms) and when all players with $j=j^*_i$ have valuations $v_i^j:\Xsij\rightarrow \R_+$, while the rest of the players have
value $0$ for any outcome. Observe that this history of play $h_i^j$ is the same as the history of play produced by contingency plan $a$ of the global mechanism,
since the deviation of player $i$ didn't change the history up till mechanism $\Mj$.
Therefore $\acj(h^j)=(\acij(h_i^j))_{i\in [n]}$ is the action profile that would 
have been played at mechanism $m$ under the contingency plan $a$.

In fact, the above sum is at least the sum of these utilities, since we also need to subtract the payments of the players with $0$ valuation to get exactly the sum of the utilities. Let $v_j$ be the valuation profile consisting of the above induced valuations
$v_i^j$ on mechanism $\Mj$.

The value of the optimal outcome in such a setting for mechanism $\Mj$ is at least the value of outcome $x_j^*$. Hence, the smoothness of mechanism $\Mj$ says that there must exist a strategy $\al_{ij}^* = \al_{ij}^*(v_j,a_i^j(h_i^j))$ such that:
\begin{equation*}
\sum_i u_i(\al_{ij}^*,a_{-i}^j(h_{-i}^j))\geq \lambda \sum_{i: j^*_i=j}v_i^j(x_{ij}^*) - (1+\mu)\sum_i P_i^j(\acj(h^j))
\end{equation*}
Since the utilities of the agents with $0$ valuation never help the left hand side of the above sum, smoothness actually implies that there exist strategies $\al_{ij}^*=\al_{ij}^*(v_j,a_i^j(h_i^j))$ only for the players with non-zero valuation, such that the sum over the utilities of only those agents is at least the right hand side in the above equation.

Note again that $P_i^j(\acj(h^j))$ is the payment made at mechanism $j$ under strategy profile $a$ of the global game, since the deviation of the player didn't change the history of play. 

Thus if we set the randomized strategies of the players to follow the above smoothness deviation we will get the theorem by similar reasoning as in the proof of theorem \ref{thm:simultaneous}. Observe that the deviation $\al_{ij}^*(v_j,a_i^j(h_i^j))$ is a whole contingency plan: play until mechanism $j$ and then observe $h_i^j$; conditional on $h_i^j$ figure out which action you would have played under your initial strategy $a_i$; then use the smoothness deviation corresponding to this action. 
\end{proofof}

\subsection{Section \ref{SEC:WEAK}:
Weak Smoothness}

\begin{proofof}{Theorem \ref{thm:weak-efficiency}}
For the case of a correlated equilibrium $\al$ in the full information setting, using similar arguments as in Theorem \ref{thm:smooth} we can show that:
\begin{align*}
\E_{\al}[\sum_i u_i(\al)]\geq \lambda \opt(v) -& \mu_1 \E_{\al}[\sum_i P_i(\al)]\\
-&\mu_2 \E_{\al}[\sum_i B_i(\al_i,X_i(\al))]
\end{align*}
Using the no overbidding assumption and the fact that $u_i(a)=v_i(a)-P_i(a)$ we get:
\begin{align*}
(1+\mu_2)\E_{\al}[\sum_i v_i(\al)]\geq~& \lambda \opt(v) - (\mu_1-1) \E_{\al}[\sum_i P_i(\al)]
\end{align*}
The result now follows if $\mu_1 \le 1$. When $\mu_1 >1$, we use the fact that players have the possibility to withdraw from the mechanism and get 0 utility, we have that $\E_{\al}[v_i(X_i(\al))]\ge \E_{\al}[P_i(\al)]$, and we get the result.

For the incomplete information setting, using similar arguments as in Theorem \ref{thm:extension-theorem} we can conclude that:
\begin{multline*}
\E_v\left[\sum_i u_i^{v_i}(s(v))\right]\geq \lambda \E_w\left[\opt(w)\right] - \mu_1 \E_v\left[\sum_i P_i(s(v))\right]\\-
\mu_2\E_v\left[\sum_i B_i(s_i(v_i),X_i(s(v)))\right]
\end{multline*}
The result follows by using the no-overbidding assumption and the quasi-linearity of utilities, similar to the complete information case.
\end{proofof}

\begin{theorem}\label{thm:weak-simultaneous}
Consider the mechanism defined by the simultaneous composition of $m$ mechanisms. Suppose that each mechanism $j$ is weakly $(\lambda,\mu_1,\mu_2)$-smooth when the mechanism restricted valuations of the players come from a class of valuations $(\Vij)_{i\in N}$. If the valuation $v_i:\Xsi\rightarrow \R^+$ of each player across mechanisms is XOS and can be expressed by induced valuations $v_{ij}^l\in \Vij$ then the global mechanism is also weakly $(\lambda,\mu_1,\mu_2)$-smooth.
\end{theorem}

\begin{theorem}\label{thm:weak-sequential}
Consider the mechanism defined by the sequential composition of $m$ mechanisms. Suppose that each mechanism $j$ is weakly $(\lambda,\mu_1,\mu_2)$-smooth when the mechanism restricted valuations of the players
come from a class of valuations $(\Vij)_{i\in N}$. If the valuation $v_i:\Xsi\rightarrow \R^+$ of each player across mechanisms is unit-demand with $v_{ij}\in \Vij$ then the global mechanism is weakly $(\lambda,\mu_1+1,\mu_2)$-smooth.
\end{theorem}

\begin{proofof}{Theorems \ref{thm:weak-simultaneous} and \ref{thm:weak-sequential}}The proofs of the above two theorems are identical to the proofs of Theorems \ref{thm:simultaneous} and \ref{thm:sequentially} combined with the following extra argument: Since action spaces and outcome decisions at a mechanism are independent of those in other mechanisms, the willingness-to-pay of a player is additive: 
\begin{align*}
\sum_j B_i^j(\acij,\xsij)=~&\sum_j \max_{a_{-i}^j:X_i^j(\acj)=\xsij} P_i^j(\acj)\\
 =~& \max_{a_{-i}: X_i(a)=\xsi}\sum_j P_i^j(\acj)\\
 =~& \max_{a_{-i}:X_i(a)=\xsi} P_i(\ac) 
= B_i(\aci,\xsi)
\end{align*}
\end{proofof}

\subsection{Section \ref{SEC:BUDGET}: Budget Constraints}

\begin{proofof}{Theorem \ref{thm:conservative-efficiency}}
We begin with the following observation: if for any action profile $a$ in the support of 
a random action profile $\al$ it holds that $P_i(a)<B_i$ then the expected utility of a player
with a budget constraint $B_i$ is the same as the expected utility of an unconstrained player
with quasi-linear utilities. 

Consider a valuation and budget profile $t=(v,B)$ and let $u_i^{t_i}$ denote the utility of player
$i$ with type $t_i=(v_i,B_i)$. Let $\hat{v}$ be the corresponding capped valuation profile
where each players value is replaced with $\hat{v}_i=\min\{v_i,B_i\}$ and $\hat{u}_i^{\hat{v}_i}$ be
a quasi-linear utility with valuation $\hat{v}_i$. 

Since we assumed that the mechanism is smooth in the quasilinear setting and the valuation space is closed under capping, for any strategy profile $a$ there exists a randomized strategy $\al_i^*(\hat{v},a_i)$ for each player such that under the quasilinear utility setting:
\begin{equation}
\sum_i \hat{u}_i^{\hat{v}_i}(\al_i^*(\hat{v},a_i),a_{-i})\geq \lambda \opt(\hat{v})- \mu \sum_i P_i(a)
\end{equation}
and such that for all $a_i^*$ in the support of $\al_i^*(\hat{v},a_i)$:
\begin{equation}\label{eqn:conservative-bid}
\max_{\tilde{a}_{-i}} P_i(a_i^*,\tilde{a}_{-i})\leq \max_{x_i\in \X_i} \hat{v}_i(x_i)\leq B_i
\end{equation}

Suppose that in the budgeted setting each player $i$ deviates to $\al_i^*(\hat{v}_i,a_i)$. Then 
by the above conservativeness of this deviating strategy and the initial observation we know that 
the expected utility of a budgeted player under this deviation is the same as the expected
utility of a player with quasi-linear utilities and value $v_i$. Subsequently the expected
utility of a player with quasi-linear utilities and value $v_i$ is at least the utility
of a player with quasi-linear utilities and value $\hat{v}_i$. Thus we get:
\begin{align}
\sum_i u_i^{t_i}(\al_i^*(\hat{v},a_i),a_{-i}) =~& \sum_i \hat{u}_i^{v_i}(\al_i^*(\hat{v},a_i),a_{-i})\nonumber\\
\geq~& \sum_i \hat{u}_i^{\hat{v}_i}(\al_i^*(\hat{v},a_i),a_{-i})\nonumber\\
\geq~& \lambda \opt(\hat{v})- \mu \sum_i P_i(a)\label{eqn:cons-smooth-proof}
\end{align}

Using the above property we can now complete the proof of the Theorem similar to the proofs of Theorems
\ref{thm:smooth} and \ref{thm:extension-theorem}. The only extra point we need to make is that due
to the fact that a player can always drop out we know that at any equilibrium solution concept no
player is going to ever be exceeding his budget at any action profile in the support of the
solution concept since otherwise his utility would have been minus infinity. Hence, at any  action
profile in the support of an equilibrium the utility of a player will behave as if quasilinear.
For completeness we present here the two proofs. 

\textbf{Correlated Equilibrium.} Consider the full information setting and let $\al$ be a correlated equilibrium. Since player $i$ doesn't want to switch $a_i$ with $\al_i^*(\hat{v},a_i)$ we get:
\begin{equation*}
\forall a_i: \E_{\al_{-i}|a_i}\left[u_i^{t_i}(a_i,\al_{-i})\right]\geq \E_{\al_{-i}|a_i}\left[u_i^{t_i}(\al_i^*(\hat{v},a_i),\al_{-i})\right]
\end{equation*}
The theorem follows by taking expectation over $a_i$ and adding over all players:
\begin{align*}
\E_{\al}\left[\sum_i u_i^{t_i}(\al)\right]\geq~& \E_{\al}\left[\sum_i u_i^{t_i}(\al_i^*(\hat{v},\al_i),\al_{-i})\right]\\
\geq~& \lambda \opt(\hat{v}) - \mu\E_{\al}\left[\sum_i P_i(\al)\right]
\end{align*}
For any $a$ in the support of correlated equilibrium $\al$ each player must be paying less than his budget 
or otherwise his expected value will be negative and he could deviate to dropping out. Hence, for 
any $a$ in the support of $\al$ the utility of a player is quasi-linear and we have that $v_i(a)=u_i^{t_i}(a)+P_i(a)$. Hence
\begin{align*}
\E_{\al}\left[\sum_i v_i(\al)\right]\geq~& \lambda \opt(\hat{v}) +(1-\mu)\E_{\al}\left[\sum_i P_i(\al)\right].
\end{align*}
The result now follows if $\mu \le 1$. When $\mu >1$, we use the fact that players have the possibility to withdraw from the mechanism and get nonnegative utility, we have that $\E_{\al}[v_i(\al)]\ge \E_{\al}[P_i(\al)]$, and we get the result.

\textbf{Bayes-Nash Equilibrium.} For the case of Bayes-Nash equilibrium we consider the following deviation for each player $i$ that depends only on the information that he has which is his own type $t_i=(v_i,B_i)$: He random samples a type profile $\tau=(w,C)\sim F$. Let $\hat{w}_i = \min\{w_i,C_i\}$ be the capped random sampled valuations. Then he plays $\al_i^*((\hat{v}_i,\hat{w}_{-i}),s_i(\tau_i))$. 

Since this is not a profitable deviation for player $i$, it means:
\begin{align*}
\E_{t}\left[u_i^{t_i}(s(v))\right]\geq~& \E_t \E_{\tau} [u_i^{t_i}(\al_i^*((\hat{v}_i,\hat{w}_{-i}),s_i(\tau_i)),s_{-i}(t_{-i})]\\
=~& \E_t \E_\tau[u_i^{\tau_i}(\al_i^*((\hat{w}_i,\hat{w}_{-i}),s_i(t_i)),s_{-i}(t_{-i}))]\\
=~& \E_t \E_\tau[u_i^{\tau_i}(\al_i^*(\hat{w},s_i(t_i)),s_{-i}(t_{-i}))]
\end{align*}
The second equality comes from exchanging the names of the variables $t_i=(v_i,B_i)$ and $\tau_i=(w_i,C_i)$.
Variables $t_i$ and $\tau_i$ are distributed identically and independently with each other and with any other variable and this enables the latter name change and regrouping. By doing this change of variables, $\hat{v}_i=\min\{v_i,B_i\}$ becomes $\hat{w}_i=\min\{w_i,C_i\}$.

Summing over all players and using Equation \eqref{eqn:cons-smooth-proof}:
\begin{align}
\E_t\left[SW^t(s(t))\right]\geq~& \E_t \E_\tau\left[\sum_i u_i^{\tau_i}(\al_i^*(\hat{w},s_i(t_i)),s_{-i}(t_{-i}))\right]\nonumber\\
\geq~& \E_t \E_\tau\left[\lambda \opt(\hat{w}) - \mu\sum_i P_i(s(t))\right]\nonumber\\
=~& \lambda \E_\tau[\opt(\hat{w})] - \mu \E_t\left[\sum_i P_i(s(t))\right]\nonumber
\end{align}
Again due to the fact that bidders can drop out, we know that for any action in the support of strategy profile $s(t)$
a player $i$ is never paying above his budget. If $s(t)$ is a Bayes-Nash equilibrium then it must be
that the utility of a player is quasi-linear in the support of $s(t)$. The theorem then follows by this quasi-linearity of utility and by the fact that expected revenue is at most the expected social welfare.
\end{proofof}

\begin{proofof}{Theorem \ref{thm:budget-composition}}
The theorem is proved in a sequence of two lemmas presented below. 

\begin{lemma}\label{thm:comp-cons}
The simultaneous composition of $m$ conservatively $(\lambda,\mu)$-smooth mechanisms 
is also conservatively $(\lambda,\mu)$-smooth, when the valuation $v_i:\Xsi\rightarrow \R^+$ of each player across mechanisms is XOS and can be expressed by induced valuations $v_{ij}^l\in \Vij$.
\end{lemma}
\begin{proof}
We want to show that the simultaneous composition is a conservatively $(\lambda,\mu)$-smooth mechanism. This means
that we should show that it is smooth in the quasi-linear setting and that the deviations used to show it 
is smooth satisfy the property that every action $a_i$ in their support satisfies:
$$\max_{a_{-i}}P_i(a_i,a_{-i})\leq \max_{x_i \in \X_i}v_i(x_i)$$

The fact that the composition is smooth just stems from Theorem \ref{thm:simultaneous}, since
each component mechanism is conservatively smooth and thereby smooth. 

From the proof of Theorem \ref{thm:simultaneous} we know that for each action profile $a$
the deviation that is used in the smoothness argument is a randomized deviation $\al^*_i(v,a_i)$ that
consists of independent randomized deviations for each mechanism $j$ following the 
distribution of $\al_{ij}^*(v_{j}^{*},\acij)$, where $v_j^{*}$ is the 
valuation profile for mechanism $j$ where each player has valuation $v_{ij}^{*}$
on $\Xsij$ and where $v_{ij}^*$ is the additive valuation that corresponds to allocation $x_i^*$ according to the XOS definition, i.e., $v_i(x_i^*)=\sum_j v_{ij}^{*}(x_{ij}^*)$ and for all $\xsi\in \Xsi$: $v_i(x_i)\geq \sum_j v_{ij}^{*}(\xsij)$.

Now by conservative smoothness of each component mechanism $\Mj$ we know that for any
action $a_{ij}^*$ in the support of $\al_{ij}^*(v_{j}^{*},\acij)$:
\begin{equation}\label{eqn:comp-cons-1}
\max_{a_{-i}^j} P_i^j(a_{ij}^*,a_{-i}^j) \leq 
\max_{\xsij\in \Xsij} v_{ij}^{*}(\xsij)
\end{equation}
For each mechanism $\Mj$ let $\hat{x}_i^j$ be the allocation 
that corresponds to the maximizer on the right hand side of the above inequality. 

By the fact that action spaces are independent across mechanisms it is easy to see that
for any action $a_i$:
\begin{align}\label{eqn:comp-cons-2}
\max_{a_{-i}} P_i(\aci,a_{-i}) =~& \sum_{j} \max_{a_j^{-i}}P_i^j(\acij,a_j^{-i})
\end{align}

Any action $\aci$ in the support of the randomized deviation $\al_i^*(v,a_i)$ is going 
to consist of actions $\acij$ in the support of the $\al_{ij}^*(v_{j}^{*},\acij)$. By
Equations \eqref{eqn:comp-cons-1} and \eqref{eqn:comp-cons-2} and by the fact that $v_{ij}^{*}$ is part of the XOS representation of $v_i$ we get that for any $\aci$ in the support of the smoothness deviation $\al_i^*(v,a_i)$:
\begin{align*}
\max_{a_{-i}} P_i(a_i,a_{-i}) =~& \sum_{j} \max_{a_{-i}^j}P_i^j(\acij, a_{-i}^j)\\
\leq~& \sum_j v_{ij}^{*}(\hat{x}_i^j) \\
\leq~& v_i(\hat{x}_i^1,\ldots,\hat{x}_i^m)\\
\leq~& \max_{x_i\in \X_i} v_i(x_i)	
\end{align*}
The latter completes the proof.
\end{proof}

To complete the proof we show a property of capped XOS valuations across mechanism outcomes:
\begin{lemma}\label{thm:capped-XOS}
Suppose that a valuation $v:\Xsi\rightarrow \R^+$ over $m$ mechanisms is XOS and can be represented by a set of additive valuations $\El$. Then the capped valuation $\hat{v}(\xsi) = \min\{v(\xsi),B_i\}$ is also XOS and can be expressed by induced valuations $\hat{v}_j^\ell:\Xsij\rightarrow \R^+$ that are cappings of the induced 
valuations $v_j^\ell:\Xsij\rightarrow \R^+$ of the initial valuation $v$.
\end{lemma}
\begin{proof}
Let $\El$ be the set of additive valuations that are used to express valuation $v$, i.e. 
$\forall x\in \X_i: v(x) = \max_{\ell\in \El}\sum_j v_j^{\ell}(x_j)$. Let $\ell(x)$ be the 
additive valuation corresponding to outcome $x$, i.e.:
$$\ell(x) = \arg\max_{\ell} \sum_j v_j^{\ell}(x_j)$$
Consider the set of additive valuations $\hat{\El}$ that contains $\ell(x)$ for each $x\in\X_i$, potentially having
multiple copies of the same additive valuation. The additive valuation associated with each $\ell(x)\in \hat{\El}$ 
is now defined as follows:
\begin{equation}
\hat{v}_j^{\ell(x)}(\tilde{x}_j) = \min\{v_j^{\ell}(\tilde{x}_j),v_j^{\ell}(x_j),B_i-\sum_{k<j}v_k^{\ell}(x_k)\}
\end{equation}
Consider an outcome $\tilde{x}$. We will show that 
$$\hat{v}(\tilde{x})=\max_{\ell(x)\in \hat{\El}} \sum_j \hat{v}_j^{\ell(x)}(\tilde{x}_j)$$

\noindent\textbf{Case 1:} If $v(\tilde{x})\leq B_i$ then:
\begin{align}\label{eqn:capped-xos-1}
\hat{v}(\tilde{x})=~&v(\tilde{x}) = \sum_j v_j^{\ell(\tilde{x})}(\tilde{x}_j)\leq B_i
\end{align}
The last inequality implies that for all $j$:
$$\sum_{k\leq j}v_k^{\ell(\tilde{x})}(\tilde{x}_k)\leq B_i$$
which in turn implies that 
$$v_j^{\ell(\tilde{x})}(x_j)\leq B_i - \sum_{k<j}v_k^{\ell(\tilde{x})}(\tilde{x}_k)$$
Therefore, for all $j$: 
$$v_j^{\ell(\tilde{x})}(\tilde{x}_j) = \hat{v}_j^{\ell(\tilde{x})}(\tilde{x}_j)$$
Combining with Equation \eqref{eqn:capped-xos-1} we get:
\begin{align}\label{eqn:xos-capped-2}
\hat{v}(\tilde{x})= \sum_j \hat{v}_j^{\ell(\tilde{x})}(\tilde{x}_j)
\end{align}
Now we need to show that for all $\ell(x)\in \hat{\El}$: $v(x)\geq \sum_j \hat{v}_j^{\ell(x)}(\tilde{x}_j)$.
Using the XOS property of the initial representation of the uncapped value we have:
\begin{align}\label{eqn:xos-capped-3}
\hat{v}(\tilde{x})=~&v(\tilde{x})\geq \sum_j v_j^{\ell(x)}(\tilde{x}_j) \geq \sum_j \hat{v}_j^{\ell(x)}(\tilde{x}_j)
\end{align}
By Equations \eqref{eqn:xos-capped-2} and \eqref{eqn:xos-capped-3} we get that:
\begin{equation}
\hat{v}(\tilde{x}) = \max_{\ell(x)\in \hat{\El}}\sum_j\hat{v}_j^{\ell(x)}(\tilde{x}_j)
\end{equation}

\noindent\textbf{Case 2:} If $v(\tilde{x})>B_i$ then $\hat{v}(\tilde{x})=B_i$. First we observe that by the definition of the capped additive valuations $\hat{v}_j^{\ell(x)}$ we know that for any $\ell(x)\in \hat{\El}$:
\begin{align*}
\sum_j \hat{v}_j^{\ell(x)}(\tilde{x}_j) =~& \sum_j \min\{v_j^{\ell}(\tilde{x}_j),v_j^{\ell}(x_j),B_i-\sum_{k<j}v_k^{\ell}(x_k)\}\\
\leq~& \sum_j \min\{v_j^{\ell}(x_j),B_i-\sum_{k<j}v_k^{\ell}(x_k)\}\\
\leq~& B_i = \hat{v}(\tilde{x})
\end{align*}
In addition we know that:
\begin{align*}
\sum_j \hat{v}_j^{\ell(\tilde{x})}(\tilde{x}_j)=~& \sum_j \min\{v_j^{\ell}(\tilde{x}_j),v_j^{\ell}(\tilde{x}_j),B_i-\sum_{k<j}v_k^{\ell}(x_k)\}\\
=~& \sum_j \min\{v_j^{\ell}(\tilde{x}_j),B_i-\sum_{k<j}v_k^{\ell}(x_k)\}\\
=~& B_i = \hat{v}(\tilde{x})
\end{align*}
By the above two sets of equations we get again that:
\begin{equation}
\hat{v}(\tilde{x}) = \max_{\ell(x)\in \hat{\El}}\sum_j\hat{v}_j^{\ell(x)}(\tilde{x}_j)
\end{equation}

Thus we conclude that for any $\tilde{x}$, the above property holds and therefore $\hat{v}_j^{\ell(x)}$
for all $\ell(x)\in\hat{\El}$ is an XOS representation of $\hat{v}$ that uses only capped 
induced valuations of the initial representation of $v$.
\end{proof}

The above two lemmas complete the proof of the theorem by simply invoking Theorem \ref{thm:conservative-efficiency}.
\end{proofof}

\end{document}